\documentclass[11pt]{article}
\usepackage{amsmath,amsfonts,amssymb,amsthm,bm}
\usepackage[colorlinks]{hyperref}
\hypersetup{linkcolor=cyan,filecolor=cyan,citecolor=cyan,urlcolor=cyan}
\usepackage[noadjust]{cite}
\usepackage{fullpage}

\usepackage{xspace}
\usepackage{thm-restate,color,xcolor}

\usepackage[boxed]{algorithm}

\usepackage{mathtools}
\usepackage{framed}
\usepackage[framemethod=tikz]{mdframed}
\usepackage{cleveref,aliascnt}
\usepackage{tikz}
\usepackage{wrapfig}
\usepackage{framed}
\usepackage{paralist}
\usepackage{graphicx}
\usepackage{tabularx,booktabs}
\usepackage{multirow}
\usepackage{paralist}
\usepackage{titling}

\newcolumntype{C}[1]{>{\centering\let\newline\\\arraybackslash\hspace{0pt}}m{#1}}

\newcommand{\cA}{\mathcal{A}}

\DeclareMathOperator{\polylog}{polylog}
\DeclareMathOperator{\poly}{poly}

\theoremstyle{plain}
\newtheorem{theorem}{Theorem}[section]

\newtheorem{lemma}[theorem]{Lemma}

\newtheorem{corollary}[theorem]{Corollary}

\newtheorem{definition}{Definition}

\newtheorem{strategy}[theorem]{Strategy}

\newcommand{\ID}{\mathrm{id}}

\newcommand{\Flip}{\mathrm{Flip}}
\newcommand{\Prefix}{\mathrm{Prefix}}
\newcommand{\Suffix}{\mathrm{Suffix}}
\newcommand{\Target}{\mathrm{Target}}
\newcommand{\Source}{\mathrm{Source}}
\newcommand{\Decode}{\mathrm{Decode}}
\newcommand{\Send}{\mathsf{Send}}

\newcommand{\loc}{\mathsf{Pos}}
\newcommand{\LDCDecode}{\mathrm{LDCDecode}}
\newcommand{\LDCDecodeIndices}{\mathrm{DecodeIndices}}

\newcommand{\rk}{\mathrm{rank}}

\newcommand{\Minput}{M_{\mathrm{in}}}
\newcommand{\Moutput}{M_{\mathrm{out}}}
\newcommand{\Mbad}{M_{\mathrm{bad}}}

\DeclareMathOperator{\Hamm}{Hamm}

\newcommand{\abs}[1]{\left|#1\right|}
\newcommand{\sett}[2]{\left\{ #1 \left| \; \vphantom{#1 #2} \right. #2  \right\}} 

\newcommand{\stt}{\medspace | \medspace}

\newcommand{\outload}{\mathsf{OutLoad}}
\newcommand{\inload}{\mathsf{InLoad}}
\newcommand{\OUT}{\mathsf{OUTind}}
\newcommand{\IN}{\mathsf{INind}}

\newcommand{\ExtMatchingTransmission}{\mathsf{SMRoutingProtocol}}
\newcommand{\SuperMessagesRouting}{\mathsf{SuperMessagesRouting}}
\newcommand{\StaticAllToAll}{\mathsf{NonAdaptiveAlltoAll}}
\newcommand{\AdaptativeAllToAll}{\mathsf{AdaptiveAlltoAll}}

\newcommand{\AllToAllComm}{\mathsf{AllToAllComm}}

\newcommand{\Recover}{\mathrm{Recover}}
\newcommand{\Add}{\mathrm{Add}}

\newcommand{\LDC}{\mathrm{LDC}}

\newcommand{\Sk}{\mathrm{Sk}}

\begin{document}
		\title{All-to-All Communication with Mobile Edge Adversary: \\Almost Linearly More Faults, For Free}

\date{}
\author{
Orr Fischer \\
	\small Bar-Ilan University \\
	\small fischeo@biu.ac.il \\
	\and
	Merav Parter \\ 
	\small Weizmann Institute of Science \\
	\small merav.parter@weizmann.ac.il	\thanks{This project is funded by the European Research Council (ERC) under the European Union’s Horizon 2020 research and
innovation programme (grant agreement No. 949083), and by the Israeli Science Foundation (ISF), grant No. 2084/18, 1042/22 and 800/22.}			
}
    
    \maketitle
    \begin{abstract}
Resilient computation in all-to-all-communication models has attracted tremendous attention over the years. Most of these works assume the classical faulty model which restricts the total number of corrupted edges (or vertices) by some integer fault parameter $f$. A recent work by [Bodwin, Haeupler and Parter, SODA 2024] introduced a stronger notion of fault-tolerance, in the context of graph sparsification,  which restricts the degree of the failing edge set $F$, rather than its cardinality. For a subset of faulty
edges $F$, the faulty-degree $\deg(F)$ is the largest number of faults in $F$ incident to any given node.

In this work, we study the \emph{communication} aspects of this faulty model which allows us to handle almost linearly more edge faults (possibly quadratic), with no extra cost. Our end results are general compilers that take any Congested Clique algorithm and simulate it, in a round by round manner, in the presence of a $\alpha$-Byzantine mobile adversary that controls a $\alpha$-fraction of the edges incident to each node in the fully connected network. For every round $i$, the mobile adversary is allowed to select a distinct set of corrupted edges $F_i$ under the restriction that $\deg(F_i)\leq \alpha n$. In the non-adaptive setting, the $F_i$ sets are selected at the beginning of the simulation, while in the adaptive setting, these edges can be chosen based on the entire history of the protocol up to round $i$. Our main results are:

\begin{itemize}
\item{\textbf{Non-Adaptive Adversary:}} For a sufficiently small constant $\alpha \in (0,1)$, we provide a $O(1)$-round simulation of 
any (fault-free) Congested Clique round in the presence of a \emph{non-adaptive} $\alpha$-adversary. Hence, supporting a total of $\Omega(n^2)$ edge corruptions. 

\item {\textbf{Adaptive Adversary:}} Our key result is a randomized $O(1)$-round simulation of a Congested Clique round in the presence of an \emph{adaptive} $\alpha$-adversary for $\alpha=1/n^{o(1)}$ (hence, a total of $n^{2-o(1)}$ edge corruptions).
Our value of $\alpha$ critically depends on the state-of-the-art query complexity of locally decodable codes with constant rate and distance by [Kopparty, Meir, Ron{-}Zewi and Saraf, J. ACM 2017]. This improves considerable over the prior work of [Fischer and Parter, PODC 2023] that provided $\poly\log n$-round simulation of a Congested Clique round in the presence of $\Theta(n)$ mobile edge corruptions. We note that their algorithm cannot support even $\alpha=1/n$.

\item{\textbf{Deterministic Compilers:}} Finally, we provide deterministic simulations against adaptive $\alpha$-adversaries that run in $O(1)$ rounds (resp., $O(\log n)$ rounds) for $\alpha=O(1/\sqrt{n})$ (resp., $\alpha=\Theta(1)$).

\end{itemize}
A key component of our algorithms is a new resilient routing scheme which may be of independent interest. Our approach is based on a combination of techniques, including error-correcting-code, locally decodable codes, cover-free families, and sparse recovery sketches. 
\end{abstract}
		
    \tableofcontents
    \newpage
    \section{Introduction}
Resilient computation in all-to-all communication models has attracted significant attention since the early 80's due to the growing need for fault-tolerant distributed systems in large-scale and high-performance computing environments.  Among all types of faults, Byzantine faults model a worst-case fault scenario where faulty components are controlled by an all-powerful adversary that can behave arbitrarily (even maliciously) by either stopping, rerouting, or altering transmitted messages in a way most detrimental to the communication process. So-far, most of the prior work focused on the classical (static) Byzantine setting, where for a given fault parameter $f$, the adversary is assumed to control a fixed subset of at most $f$ edges (or nodes) in a fully connected network on $n$ nodes. Within this setting, time and communication efficient algorithms have been devised for various tasks in complete network topologies, supporting up to $f=\Theta(n)$ edge and node corruptions. Examples include: broadcast and consensus \cite{PSL80,D82,DFFLS82,DFFLS82,F83,BT85,B87,TKS87,SW89,BGP89,SW90,CW92,BG93,FM97,GM98,FM20,Ko04,PP05,KK06,DH08,MT12,IR15,CHMOS19,KNV19}, gossiping \cite{BP93,BH94,CT17}, other reliable transmission tasks \cite{DPPU88,U94,JRV20,BM24,CGO15,DDWY93,KS09,ACH06}, and additional tasks \cite{AMPV22,MR23}.

In their seminal work \cite{OstrovskyY91}, Ostrovsky and Yung introduced the study of distributed computation in the presence of mobile Byzantine adversaries which can dynamically appear throughout the network, analogous to a spread of a virus. Specifically, in the mobile setting for a given fault parameter $f$, in every round $i$ of the protocol, the mobile adversary is allowed to control a possibly distinct subset $F_i$ of at most $f$ edges (or nodes). Mobile Byzantine (node) faults have been also addressed by Garay \cite{G94} in the context of the byzantine agreement problem. Tight bounds for this problem, in terms of the allowed number of faults per round, have been provided by Bonnet et al. \cite{BDNP16}. See Yung \cite{Y15} for an overview on mobile adversaries. Our work is mostly related to the recent work by Fischer and Parter \cite{FP23} who presented a simulation of a (fault-free) Congested Clique round by $\polylog n$ rounds in the mobile Byzantine \emph{edge} adversary with $f=\widetilde{\Theta}(n)$. 

\smallskip
\noindent \textbf{New: Mobile Byzantine Adversary with Bounded Faulty Degree.}  We study a considerably stronger notion of a (mobile) Byzantine edge adversary
-- compared, for instance, to the recent work of \cite{FP23} -- where the constraint is placed on the \emph{degree} of the faulty edge set rather than its size.
For $\alpha \in (0,1)$, an $\alpha$-Byzantine Degree (BD) adversary manipulates a subset of edges $F_i$ in each round $i$ ensuring that no node has more than an $\alpha$-fraction of its incident edges in $F_i$. This model can be viewed as a mobile generalization of the bounded degree fault-tolerant model by Pelc and Peleg \cite{PP05}. Our work is in particular inspired by the recent work of Bodwin, Haeupler and Parter \cite{BHP24} that studied bounded degree faults in the context of fault-tolerant (FT) graph sparsification. Their key results provide new graph structures with size comparable to previous constructions (in the classical FT model), but which tolerate every fault edge set $F$ of small faulty degree $\deg(F)$, rather than only fault edge sets of small size $|F|$. This effectively enables resilience to a linear factor of more faults at little to no additional cost.

Our goal is to achieve a similar impact in the context of resilient distributed computation where the key complexity measure is the number of communication rounds of protocols (rather than the size of a graph structure, as in \cite{BHP24}). As we will see (and also observed in the work of \cite{BHP24}), this setting calls for a new approach compared to prior works where the total number of corrupted edges was limited to be at most linear. Our communication model is the classical Congested Clique, introduced by Lotker et al. [22]. In this model, the $n$ nodes are fully connected and communication occurs in synchronous rounds. In every round, each node can send $B$ (possibly distinct) bits to each of the nodes in the network; typically $B=O(\log n)$. We distinguish between two variants of $\alpha$-BD adversaries. The \emph{non-adaptive} adversary, denoted by $\alpha$-NBD, must select the $F_i$ edge sets (for every round $i$) at the beginning of the simulation. In contrast, the adaptive adversary, 
denoted by $\alpha$-ABD, is allowed to choose the identity of the $F_i$ edges at the beginning of each round $i$, based on the entire history of the protocol up to round $i$.
Our starting point is the prior work of Fischer and Parter \cite{FP23} that provides general compilers for Congested Clique algorithms against an adaptive adversary that controls $\widetilde{\Theta}(n)$ edges in each round with a $\polylog n$-round overhead. A critical component in their algorithm is an upcast procedure over $\Theta(n)$ edge disjoint trees. This approach totally breaks in the bounded degree setting, as a faulty set of edges forming a matching (i.e., $\alpha=1/n$) can destroy the entire collection of their edge disjoint trees.  Moreover, their round complexity is inherently poly-logarithmic while we strive for supporting $\alpha=\Theta(1)$ with a constant round complexity. 

\noindent \textbf{Additional Related Work.}
Several prior works consider adversarial models where the adversary controls an $\alpha$-fraction of the edges, for some corruption parameter $\alpha > 0$ \cite{BM24,CGO15}. In such models, the reliable network may become disconnected, hence the tasks considered deal with outputs of all but some fraction of the nodes, e.g. \emph{almost-everywhere} agreement. Both works perform these tasks in specifically crafted networks, with applications to the full clique, and are part of a larger line of works relating to \emph{almost-everywhere} reliable transmission \cite{DPPU88,U94,JRV20,BM24,CGO15}. 
Hoza and Schulman \cite{HS16,S92} showed that in a setting where the adversary can corrupt an $\alpha$-fraction of the communication bits sent across a network, any protocol can be simulated resiliently as long as $\alpha = O(1/n)$. Romashchenko \cite{R06} showed that in some shared memory model, a set of $2^{O(n)}$ processors, where at most a constant fraction of them are corrupted in an adversarial manner, can efficiently simulate the computation of $n$ non-faulty processors, via the use of locally decodable codes. Another extensively studied topic is resilient computation in a \emph{stochastic} adversarial model \cite{DP90,P92,RS94,ABEGH19, GMS11,GMS14}, where e.g. messages are randomly corrupted independently with some probability $p$.

    \smallskip
\noindent\textbf{Our Results.}
All our algorithms are implemented in the $B$-Congested Clique model where $B$ is the allowed message size to be sent on a given edge in each round. 
Most of our results hold already for the more challenging setting $B=1$ (hence, immediately also hold for $B=O(\log n)$).
The local computation time of the nodes is \emph{polynomial}\footnote{While the classical model do not account for the local computation time, we make extra efforts in several of our algorithms, to keep the local computation polynomial in $n$.}. 

\smallskip
\noindent\emph{Adversarial Routing Lemma.} Our Congested Clique compilers are heavily based on a new routing scheme, which 
we treat as the resilient analogue for the well-known Lenzen's routing \cite{L13}, adopted to the adversarial $\alpha$-BD setting.  We believe that this may be of independent interest.

Our routing scheme formulates the conditions under which the routing can be done in $O(1)$ rounds in the presence of $\alpha$-ABD adversary. Specifically, we show that can be done provided  that each node holds $O(1/\alpha)$ messages, denoted as super-messages\footnote{We call such a message a super-message to distinguish from a $B$-bit message that can be sent over an edge in a single Congested Clique round.}, each of $O(\alpha n)$ bits; and each node is a target of $O(1/\alpha)$ super-messages. A given super-message is allowed to have multiple targets. 
By using the tools of Error Correcting Codes and specialized families of cover free sets (that we design for this setting), we show the following:
\begin{theorem}[Informal]
    \label{thm:informal_routing}
    Let $c \in (0,1)$ be a small constant. For any $\alpha \leq c$, there is a deterministic $O(1)$-round algorithm for solving the following routing instances in the presence of $\alpha$-ABD adversary: Each node is the source and target of at most $k = O(\min(1/\alpha,n/\log n))$ super-messages, each super-message is of size $O(n/k)$ bits (allowing messages to have multiple targets). We assume that the target set of each of the $kn$ super-messages is known to all the nodes.
\end{theorem}

In Section \ref{sec:overview} we provide a comparison to the fault-free routing schemes in the Congested Clique model by \cite{L13} and \cite{CFG20}. 

\smallskip
\noindent \emph{Single Round Simulation via All-to-All Communication.}
Our goal is to ``correctly'' simulate a general protocol in the (fault-free) Congested Clique model. For this purpose, it suffices to show how to simulate a single Congested Clique round in our settings. We note that we can simulate both deterministic and randomized protocols of the non-faulty Congested Clique in this manner: if we wish to simulate a randomized protocol $\mathcal{A}$ in the non-faulty Congested Clique, one can fix the randomness $R_{\mathrm{\mathcal{A}}}$ used by $\mathcal{A}$, making $\mathcal{A}$ ``deterministic'' for the purpose of the simulation (we however, do not fix the randomness used for the simulation itself). The single-round simulation task boils that into the following routing problem:

\begin{definition}
    In the $\AllToAllComm$ problem, each node $u$ is given as input $n$ messages $\{m_{u,v}\}_{v \in V}$ each of size $B$ bits, and our goal is for every node $v$ to learn all messages $\{m_{u,v}\}_{u \in V}$.
\end{definition}

Our key contribution is in providing time-efficient resilient algorithms for the $\AllToAllComm$ problem with input $B$-bit messages in the $B$-Congested Clique model. This immediately yields general compilers. Specifically, an $r$-round Congested Clique algorithm for the $\AllToAllComm$ problem provides a compiler for simulating any fault-free $r'$-round Congested Clique algorithm in the $\alpha$-BD setting in $O(r' \cdot r)$ Congested Clique rounds. Our randomized algorithms for the $\AllToAllComm$ problem guarantee that each node $v$ learns its messages $\{m_{u,v}\}_{v \in V}$ with high probability of $1-1/n^c$, for any desired constant $c\geq 1$.

We first consider the simpler non-adaptive setting, for which we obtain $O(1)$-round solution for the $\AllToAllComm$ problem which can tolerate any sufficiently small constant $\alpha=\Theta(1)$, i.e. it is essentially optimal in both round complexity and resilience to faults.

\begin{theorem}
\label{thm:informal_non_adaptive}[Randomized Compilers against Non-Adaptive Adversary]
For a sufficiently small $\alpha \in (0,1)$ and assuming $B=\Theta(\log n)$, there is 
a randomized $O(1)$-round algorithm for solving the $\AllToAllComm$ problem (i.e., single-round simulation) in the presence of $\alpha$-NBD adversary. 
\end{theorem}

Our main efforts are devoted to the adaptive setting, for which the approach taken in the algorithm of \Cref{thm:informal_non_adaptive} totally breaks. We combine the tools of Locally Decodable Codes (LDCs) and Sparse Recovery Sketches to obtain the following: 

\begin{theorem}
\label{thm:informal_adaptive_randomized}[Randomized Compilers against Adaptive Adversary]
There is a randomized $O(1)$-round algorithm for solving the $\AllToAllComm$ problem in the presence of an $\alpha$-ABD adversary with
$\alpha=\exp(-\sqrt{\log{n}\log\log{n}})$ and bandwidth $B =1$.
\end{theorem}

Our value of $\alpha$ is determined by the query complexity of LDCs with constant rate and distance, and are in particular based on the construction of \cite{KMRS17}. Compared with the algorithm of \cite{FP23}, the algorithm of \Cref{thm:informal_adaptive_randomized} is faster, and can support $n^{2-o(1)}$ corrupted edges, compared to \cite{FP23} which can support $\widetilde{\Theta}(n)$ corrupted edges. However, it does not fully subsume it. For example, the algorithm in \cite{FP23} is resilient to a case where the adversary corrupts $\widetilde{\Theta}(n)$ edges of a single node. 

\smallskip
\noindent \emph{Deterministic Solutions.} Finally, we also provide deterministic solutions for the problem for $\alpha=\Omega(1)$ and $\alpha=1/\sqrt{n}$, respectively. We note that in the deterministic setting, adaptive and non-adaptive adversaries have the same power, hence we only consider the adaptive setting.

\begin{theorem}
\label{thm:informal_adaptive_deterministic_high_noise}
For a sufficiently small constant $\alpha=\Omega(1)$ and $B=1$, there 
is a deterministic $O(\log n)$-round algorithm for the $\AllToAllComm$ problem in the presence of $\alpha$-ABD setting.
\end{theorem}

In particular this result directly improves the result of \cite{FP23} for the Congested Clique - it has improved round complexity, is resilient to $\Theta(n^2)$ corrupted edges, and is deterministic.

\begin{theorem}
\label{thm:informal_adaptive_deterministic_low_noise}
There is a deterministic $O(1)$ round-algorithm for the $\AllToAllComm$ problem in the presence of an $\alpha$-ABD adversary with
$\alpha=\Theta(1/\sqrt{n})$ and bandwidth $B =1$.
\end{theorem}

The algorithm of \Cref{thm:informal_adaptive_deterministic_low_noise} is very similar to the prior algorithm of \cite{ABEGH19} (Theorem 1) for the stochastic interactive coding setting, adapted to our model. Similar to \Cref{thm:informal_adaptive_randomized}, it is faster than that of \cite{FP23}, and can support $\Theta(n^{1.5})$ corrupted edges, but it does not fully subsume \cite{FP23}. See Table \ref{fig:table} for a summary of our results.

\begin{table}[H]
\centering
\bgroup
\def\arraystretch{1.5}
\begin{tabular}{llllll}
Fraction of Faults ($\alpha$)                                         & Adaptivity   & Randomness    & Bandwidth  & Rounds  & Reference \\ \hline 
$\Theta(1)$                                      & Non-adaptive & Randomized    & $B = \Omega(\log{n})$ & $O(1)$  & \Cref{thm:informal_non_adaptive}     \\ \hline
$\exp(-\sqrt{\log{n}\log\log{n}})$    & Adaptive     & Randomized    & $B \in \{1,\dots,\poly{n}\}$ & $O(1)$    &  \Cref{thm:informal_adaptive_randomized}   \\ \hline
$\Theta(1)$                                & Adaptive     & Deterministic & $B \in \{1,\dots,\poly{n}\}$ & $O(\log{n})$ & \Cref{thm:informal_adaptive_deterministic_high_noise}    \\ \hline
$\Theta(1/\sqrt{n})$                           & Adaptive     & Deterministic & $B \in \{1,\dots,\poly{n}\}$ & $O(1)$  & \Cref{thm:informal_adaptive_deterministic_low_noise}    \\ \hline
\end{tabular}
\egroup 
\caption{New upper bounds on the single round simulation problem in the various settings.}
\label{fig:table}
\end{table}

\textbf{Roadmap.} In \Cref{sec:key_routing} we present and prove our key routing procedure for the adaptive adversarial setting, i.e. \Cref{thm:informal_routing}. In \Cref{sec:nonadaptive} we study the non-adaptive setting, and show the construction for Theorem \ref{thm:informal_non_adaptive}. In \Cref{sec:adaptive} we consider adaptive adversaries, and show \Cref{thm:informal_adaptive_randomized}. In \Cref{sec:deterministic_compilers}, we consider the deterministic setting, and show Theorems~\ref{thm:informal_adaptive_deterministic_high_noise},\ref{thm:informal_adaptive_deterministic_low_noise}.
    \section{Preliminaries}
\label{sec:prelims}

Throughout we assume the standard Congested Clique model with bandwidth $B$.  We also assume the standard KT1 model, which allows us to assume, w.l.o.g., that all node IDs are in $\{1,\ldots, n\}$.

\smallskip
\noindent \textbf{The Adversarial Model.}
We consider an all-powerful computationally unbounded adversary, also known as Byzantine, that can corrupt messages by controlling a subset of \emph{edges} $F_i$ in each round $i$. The adversary is \emph{mobile} in the sense that the set of the edges that the adversary controls in round $i$, denoted by $F_i \subseteq E$, may vary from round to round. The mobile adversary is restricted to the bandwidth limitation of the Congested Clique model: In every round $i$, the adversary may send $B$-bit messages over all the controlled edges $F_i$, where the content of the message is determined by the adversary based on its allowed knowledge. The identity of the edges $F_i$ is unknown to the nodes, e.g., a node $v$ does not know which of its incident edges are in $F_i$.

Within this setting, we consider two adversarial variants: non-adaptive vs. adaptive. In the \emph{non-adaptive setting}, the identity of the edges $F_i$ for every round $i$ are decided at the beginning of the simulation and may depend only on the index $i$ and the initial states of the nodes (randomness excluded). The (corrupted) content of the messages exchanged over the $F_i$ edges is allowed to depend on the communication history of the network up to round $i-1$ (included) and the messages the nodes are intended to send in round $i$.\footnote{We note that the non-adaptivity refers only to the choice of $F_i$. We allow the content of the corrupted messages to be chosen in an adaptive manner.} 
In contrast, in the stronger \emph{adaptive} setting, the selection of the edges $F_i$ as well as the corrupted messages to be sent over these $F_i$ edges in round $i$ may depend on the following knowledge: the local states of all nodes in the network, including the internal randomness generated thus far (round $i$ included) at each node $v$, and all the messages sent throughout the network in rounds $1,\ldots, i-1$ and those that are intended to be sent by the nodes in round $i$.\footnote{In the literature, this is referred to as a \emph{rushing} adaptive adversary (see e.g. \cite{BH92}).}

We limit the power of the adversary by restricting the number (or fraction) of corrupted edges incident to each node, namely, the \emph{corrupted degree}. An $\alpha$-Byzantine-Degree (BD) adversary is restricted to control in each round $i$ a subset $F_i$ such that $\deg(F_i)\leq \alpha \cdot n$. Let $\alpha$-NBD (resp., $\alpha$-ABD) be an  $\alpha$-BD adversary in the non-adaptive (resp., adaptive) setting.

\smallskip
\noindent \textbf{Error correcting Codes and Locally Decodable Codes.}
Our algorithms make extensive use of error correcting codes and locally decodable codes (LDCs). We need the following definitions.
\begin{definition}[Hamming Distance]
    Let $\Sigma$ be a finite set and $\ell\in\mathbb{N}$, then the distance between $x,y\in \Sigma^\ell$ is defined by
    $\Hamm(x,y) = \abs{\sett{i\in[\ell]}{x[i]\neq y[i]}}$.
\end{definition}

\begin{definition}[Error Correcting Code]\label{def:ECC}
        An error correcting code over alphabet $\Sigma$ is a function $C: \Sigma^k \rightarrow \Sigma^n$. The relative distance of the code is defined as $\delta_C = \min_{x,y \in \Sigma^k} \Hamm(C(x),C(y))/n$. The relative rate $\tau_C$ is defined as $k/n$. A decoding function for code $C$ $\Decode: \Sigma^n \rightarrow \Sigma^k$ is a function that given a word $c' \in \Sigma^n$ such that $\exists_{x \in \Sigma^k} \Hamm(C(x),c') < \delta_C \cdot n/2$, then $\Decode(c') = x$.\footnote{We remark that if $\exists x \Hamm(C(x),c') < \delta_C n/2$ then such an $x$ must be unique, hence $\Decode$ is well-defined.}
        
\end{definition}

The most useful ECC in our context is the Justesen Code \cite{J72} which is a binary code with constant rate and distance. The precise selection of $\alpha$ in \Cref{thm:informal_routing}, \Cref{thm:informal_non_adaptive} and \Cref{thm:informal_adaptive_deterministic_low_noise} depends on the distance of this code. 

\begin{lemma}[Justesen Code \cite{J72}, Theorem 1]
\label{lem:justesen_code}
For any relative rate $\tau_C \leq 1/200$, there exists an error correcting code $C$ on binary alphabet $\Sigma = \{0,1\}$ with relative distance $\delta_C > 1/10$. The encoding and decoding algorithms can be computed in polynomial time. 
\end{lemma}

Our algorithms for the adaptive setting employ a special class of ECCs known as Locally Decodable Codes. Specifically, in our applications of error correcting codes in \Cref{sec:adaptive} the message is first partitioned into small blocks, each
of which is then encoded separately. In particular, for each node $v$ it will be desired to decode
only the portion of data in which it is interested (i.e., that corresponds to the messages $\{m_{u,v}\}_{u \in V}$). This precisely fits LDCs that simultaneously provide efficient
random-access retrieval and high noise resilience by allowing the decoding of an arbitrary bit of the message from looking at only
a small number of randomly chosen codeword bits.  

\begin{definition}[Locally Decodable Codes (LDC)]
\label{def:LDC}
An error correcting code $C: \Sigma^k \rightarrow \Sigma^n$ is called a $(q,\delta_C,\epsilon)$-LDC if there exists a randomized algorithm $\LDCDecode$ that receives as input a string $x\in \Sigma^n$, an index $i\in[k]$, and a uniformly random string $R$, performs at most $q$ queries to $x$, and has the following guarantees: if there exists $y \in \Sigma^k$ such that $\Hamm(C(y),x) \leq \delta_C n/2$, then $\Pr(\LDCDecode(x,i,R) = y[i]) \geq 1-\epsilon$, where the probability is taken over the choice of $R$. 
\end{definition}

We say an LDC is \emph{non-adaptive} if the decoding algorithm's queries only depend on the decoded index $i$ and the randomness of the algorithm. For a non-adaptive LDC, we define the function $\LDCDecodeIndices(i,R)$, which receives as input an index $i \in [k]$ and randomness $R$, and outputs the $q$ indices $x_1,\dots,x_q \in [n]$ of the codeword that local decoding algorithm $\LDCDecode$ queries given $(i,R)$ as the decoded index and randomness respectively.

\begin{lemma}[\cite{KMRS17}]
\label{lem:LDC}
For some $\delta_C = \Omega(1)$, there exists a non-adaptive $(q,\delta_C,\frac{1}{\poly(n)})$-$\LDC$ $C: \mathbb{F}_2^k \to \mathbb{F}_2^n$, for $q=\exp(\sqrt{\log n \log\log n})$ with constant relative rate. The encoding and local decoding algorithms can be computed in $\poly(n)$ time, and the local decoding algorithm uses at most $\exp(\sqrt{\log n \log\log n})$ random bits.
\end{lemma}

\noindent \textbf{Sparse Recovery Sketches.} Intuitively, a $k$-sparse recovery sketch is a succinct dynamic data structure of $\widetilde{O}(k)$ bits, where we may insert or remove elements, and we may recover all elements currently present in the data structure assuming there at most $k$ elements at the time of recovery.\footnote{More accurately, we do not insert or remove elements, but rather change their frequency in the data structure in either a positive or negative manner. The sketch can return all non-zero frequency elements, assuming there are at most $k$ such elements.} 

We refer to \cite{CJ19} for a survey. To introduce this notion more formally, consider the following setting. Let $U$ be a universe of elements, and let $\sigma$ be a multi-set of $K$ tuples $\{(e_i,f_i)\}_{i=1}^K$, where $e_i \in U$ is an element in the universe and $f_i$ is an integer referred to as the change in frequency (which can be either negative, positive, zero). The \emph{frequency} of an element $e \in U$ in the multi-set $\sigma$ is defined as the sum of its changes of frequency, i.e., $f(e) = \sum_{i:e_i = e} f_i$. Denote by $N(\sigma)$ the set of elements in $U$ with non-zero frequency in $\sigma$, i.e., $N(\sigma) = \{e \in U \stt f(e) \neq 0\}$.


\begin{lemma}[\cite{CF14}, Section 2.3.2]
    \label{lem:k_sparse_recover} 
    For an integer parameter $k \geq 1$ and a universe set $U$, there exists an algorithm $L$, which given  multi-set $\sigma$ with elements in $U \times \{-\poly(|U|),\dots,\poly(|U|)\}$ of size $|\sigma| = O(\poly(|U|))$, and a string $R$ of $O(k\log^2 |U|)$ random bits, outputs a randomized string $L(\sigma,R)$ called a $k$-recovery sketch, of size $O(k\log^2{|U|})$ bits. The following operations are defined on the sketch:
    \begin{itemize}
        \item $\Recover$: a deterministic procedure in which given a (randomized) sketch $\tau = L(\sigma,R)$ outputs $N(\sigma)$, assuming $|N(\sigma)| \leq k$, with probability $1-\frac{1}{\poly(|U|)}$ over the choice of $R$.
        \item $\Add$: a deterministic procedure in which given a sketch $L(\sigma,R)$, a pair $(e_i,f_i) \in U \times \{-\poly(|U|),\dots,\poly(|U|)\}$, and the randomness $R$, outputs $L(\sigma \cup \{(e_i,f_i)\},R)$.
    \end{itemize}
\end{lemma}

\smallskip
\noindent \textbf{Using Sparse Recovery Sketches to Correct Faulty Messages \cite{FP23}:}
Consider a single round procedure where each node $u$ sends the messages $m_{u,v}$ to each $v$, and let $\widetilde{m}_{u,v}$ be the message received by $v$.
Let $M_{bad}(v)=\{\widetilde{m}_{u,v} ~\mid~ \widetilde{m}_{u,v}\neq m_{u,v}\}$. Let $k=\max_{v \in V}|M_{bad}(v)|$ and define 
$\Sk(v)$ to be a $\Theta(k)$-sparse recovery sketch obtained by initializing an empty sketch, and applying for each $u \in V$ an $\Add(\ID(u,v) \circ m_{u,v},1,R)$ operation (i.e. adding all messages targeted at $v$, with their associated ids, with frequency $1$), where $R$ is a string of $O(k\log^2 n)$ bits.

\begin{lemma}\label{lem:sparse}
    Given $\Sk(v)$ and the randomness $R$, each $v$ can recover all messages in $M_{bad}(v)$ and their correction, i.e., obtaining $\{m_{u,v}\}_{u \in V}$.
\end{lemma}
\begin{proof}
       Given $\Sk(v)$ and $R$, node $v$ applies for each $u \in V$ the operation $\Add(\ID(u,v) \circ \widetilde{m}_{u,v},-1,R)$ (i.e. reduces the frequency of each of $v$'s received messages, with their associated ids, by $1$). We observe that the only remaining elements in $\Sk(v)$ are messages $\ID(u,v) \circ m_{u,v}$ and $\ID(u,v) \circ \widetilde{m}_{u,v}$ where $\widetilde{m}_{u,v} \neq m_{u,v}$, which is exactly the set of bad edges $\Mbad$ and their correction. By assumption, $|\Mbad(v)| \leq  k$, hence the number of elements in $\Sk(v)$ with non-zero frequency is at most $2k$. By applying $\Recover$ on $\Sk(v)$ and $R$, node $v$ recovers $\Mbad$ and their correction, corrects them, and outputs $\{m_{u,v}\}_{u \in V}$. 
\end{proof}

\smallskip
\noindent \textbf{Message Sets and Concatenation.} For node sets $A,B \subseteq V$, denote the set of messages $M(A,B) = \{m_{u,v} \mid u \in A, v \in B\}$. Unless stated otherwise, the ID of a message $\ID(m_{u,v})=\ID(u)\circ \ID(v)$. Let $\bar{M}(A,B)$
 be the set of messages in $M(A,B)$ ordered in increasing order of their IDs. For, $\bar{M}(A,B)=(m_1,\ldots, m_\ell)$, let:
 \begin{equation}\label{eq:message-concatentation}
 M^{\circ}(A,B)=m_1 \circ m_2 \circ \ldots \circ m_\ell~.  \end{equation}

For a subset of vertices $A \subseteq V$, let $A[i]$ be the $i^{th}$ node in $A$ based on increasing order of IDs.

\noindent \textbf{$k$-Wise Independent Hash Functions.}
$k$-wise independent hash functions are a powerful tool for reducing the amount of randomness required in algorithms. In essence, it is a hash function family such that a random function from the family sends any given $k$ elements $x_1,\dots,x_k$ in its input domain to a uniformly random element in its range, independently for all $k$ elements.

\begin{definition}[$k$-wise independent hash family; Definition 3.31 in \cite{V12}]
    For integers $N,k$, a family of functions $\mathcal{F} \subseteq \{f:V \rightarrow [N]\}$ is called a $k$-wise independent hash function family (with range $N$) if for any $k$ nodes $u_1,\dots,u_k \in V$, and $k$ values $x_1,\dots,x_k \in [N]$, it holds that $\Pr_{f \in \mathcal{F}}(f(u_1) = x_1 \land \dots \land f(u_k) = x_k) = \frac{1}{N^k}$. In particular, $f(u_1),\dots,f(u_k)$ are independent random variables.
\end{definition}

\begin{lemma}[Construction 3.32 in \cite{V12}]
\label{lem:kwise_indepedent_hash}
For any integers $N,k$, there exists a $k$-wise independent hash function family $\mathcal{F}$ (with range $N$) such that sampling a uniformly random function $f \in \mathcal{F}$ can be done efficiently and using $O(k\log{N})$ random bits.
\end{lemma}

\begin{lemma}[Concentration bounds for $k$-wise independent variables, Lemma 2.2, 2.3 in \cite{BR94}]
    For an even integer $k \geq 4$, let $X_1,\dots,X_m$ be $k$-wise independent random variables with values $X_i \in [0,1]$, and let $X = \sum_{i=1}^m X_i$. Denote $\mu = E(X)$. Then for any $\Delta > 0$, 

    \[\Pr(|X-\mu| \geq \Delta) \leq 8\left(\frac{k\mu + k^2}{\Delta^2} \right)^{k/2}.\]
\end{lemma}

\begin{corollary}
    \label{lem:concentration_k_wise}
    For any constant $c \geq 1$, there exist a constant $c' > 0$, such that for the setting above, for $k = \lceil c' \log{m} \rceil$, the following bounds holds:
    If $\mu \geq 16\lceil c' \log{m} \rceil$, then
    $\Pr(|X-\mu| \geq \mu/2) \leq \frac{1}{m^c},$
    and if $\mu \leq \lceil c' \log{m} \rceil$, then 
    $\Pr(|X-\mu| \geq 2 \lceil c'\log{m} \rceil) \leq \frac{1}{m^c}.$
\end{corollary}
\begin{proof}
    Taking $c' = 100\log(c+1)$, if $\mu \geq 16\lceil c'\log{m} \rceil = 16k$, then
    \[\Pr(|X-\mu| \geq 2\mu) \leq 8\left(\frac{k\mu + k^2}{\mu^2/4} \right)^{k/2} \leq 8\left(\frac{1}{2}\right)^{k/2} \leq \frac{1}{m^c}.\]
    and if $\mu \leq \lceil c'\log{m} \rceil = k$, then 
    \[\Pr(|X-\mu| \geq 2 \lceil c'\log{m} \rceil) = \Pr(|X-\mu| \geq 2k) \leq 8\left(\frac{k\mu + k^2}{4k^2} \right)^{k/2} \leq 8\left(\frac{1}{4}\right)^{k/2} \leq \frac{1}{m^c}.\]
\end{proof}

\smallskip
\noindent \textbf{Useful Assumptions.} In some of our algorithms, it is convenient to assume some properties on the value of $n$ (e.g. that $n$ is divisible by some number, $n$ is a power of $2$, or that $\sqrt{n}$ is an integer). The following lemma shows that this can be done up to a constant loss in the $\alpha$ parameter.

\begin{lemma}
\label{lem:change_n_value}
Let $n/2 \leq n' \leq n$ be an integer. If there exists a protocol $P'$ that solves $\AllToAllComm$ in the $n'$-clique for the $\alpha$-ABD setting (resp. $\alpha$-NBD setting) within $O(T)$ rounds, then there exists a protocol $P$ that solves $\AllToAllComm$ in the $n$-clique in the $(\alpha/2)$-ABD (resp. $(\alpha/2)$-NBD setting) setting within $O(T)$ rounds. The local computation time and usage of randomness of $P$ only increase by a constant factor compared to $P'$.
\end{lemma}
\begin{proof}
    First, we notice that there exists sets $V_1,\dots,V_{10} \subseteq V$ of size $|V_i| = n'$, such that any pair $u,v \in V$ is contained by at least one $V_i$. Indeed, let $$S_1 = \{v_1,\dots,v_{\lfloor n/5 \rfloor}\}, S_2 = \{v_{\lfloor n/5 \rfloor+1},\dots,v_{\lfloor 2n/5 \rfloor}\},\dots,S_5 = \{v_{\lfloor 4n/5 \rfloor+1},\dots,v_{n}\}.$$ 
    
    For each $1 \leq j < k \leq 5$ let $U_{j,k} \subseteq V$ be an arbitrary set of $n'-|S_j \cup S_k|$ vertices such that $U_{j,k} \cap (S_j \cup S_k) = \emptyset$. Assign each $i \in [10]$ a unique pair of indexes $1 \leq j_i < k_i \leq 5$, and set $V_i = S_{j_i} \cup S_{k_i} \cup U_{j_i,k_i}$. By construction, $|V_i| = n'$, and on the other hand, since $S_1,\dots,S_5$ are a partition of $V$, and for any pair of sets in the partition there exists a set $V_i$ containing them, each pair $u,v \in V$ is contained in some $V_i$. 
    
     The network runs $P'$ on each $V_i$ (where we set the input of each node $u \in V_i$ as $\{m_{u,v}\}_{v \in V_i}$). We note that in the $(\alpha/2)$-ABD (resp. $(\alpha/2)$-NBD setting) setting, in any round $j$, $\deg_{F_j}(u) \leq \alpha n/2 \leq \alpha n'$ for each $u \in V_i$. Hence, the algorithm $P'$ terminates correctly in this setting (w.h.p. if randomized). Recall each $u,v \in V$ is contained in some $V_i$, hence $v$ receives $m_{u,v}$ in the application of $P'$ on $V_i$.
\end{proof}

Finally, we show that any $r$-round algorithm for the $\AllToAllComm$ problem with $B=1$ provides an $r$-round algorithm for any $B' \in \{1,\dots,\poly{n}\}$.

\begin{lemma}[Informal]
\label{lem:bandwidth_reduction}
Let $B' \in \{1,\dots,\poly{n}\}$. Given an instance of the $\AllToAllComm$ problem with each message $m_{u,v}$ has $B=1$ bits, and let $\Pi$ be an $r$-round algorithm for solving this problem using $B$-bit messages in the $\alpha$-BD setting, w.h.p. Then, there is an algorithm $\Pi'$ for solving an instance of $\AllToAllComm$ problem with each message $m_{u,v}$ has $B'$ bits, such that $\Pi'$ runs in $r$ many $B'$-Congested Clique rounds, and has the same resilience as that of $\Pi$.
\end{lemma}
\begin{proof}
An instance of the $\AllToAllComm$ problem with each message $m_{u,v}$ has $B'$ bits can be viewed as $B'$ independent instances of an  $\AllToAllComm$ problem with $B=1$, where instance $i$ is restricted to the $i^{th}$ bit in the $m_{u,v}$ messages for every $u,v$. We can then run protocol $\Pi$ in parallel for each of the $B'$ instances. Since a single application of $\Pi$ sends $1$ bit messages on each edge, we can run the $B'$ applications in $r$ many $B'$-Congested Clique rounds. Since each application succeeds w.h.p., all $B'$ applications succeed w.h.p., as well.
\end{proof}

    \section{Technical Overview}\label{sec:overview}

\noindent \textbf{The Challenge and Some Intuition.} Consider the $\AllToAllComm$ problem with bandwidth $B=1$ and an adaptive adversary, namely, $\alpha$-ABD. In this problem, each node $v$ is required to receive $1$-bit messages $m_{u,v}$ from $n$ nodes $u \in V$, hence a total of $n$ bits. The first intuitive step is to encode each $m_{u,v}$ message into $C(m_{u,v})$ for some ECC $C$ with constant rate and distance. This step comes ``for free" as the message size only increases by a constant factor and has the benefit of 
tolerating a constant fraction of errors. From that point on, we focus on sending the codewords $\{C(m_{u,v})\}_{u,v \in V}$ to their destinations. Since the adversary corrupts $\alpha n$ edges incident to each node $u$, a node $u$ will need to send its $C(m_{u,v})$ message to at least $c\cdot \alpha n$ distinct nodes $A_{(u,v)}$ for a large enough constant $c$, to make sure that the
received codeword, distributively held by $A_{(u,v)}$, is close enough to the original message $C(m_{u,v})$. Clearly, the issue with this approach is bandwidth limitation as each node needs to communicate with $\Omega(\alpha n)$ nodes for each of its $n$ input messages. Our starting observation is that our task becomes significantly simpler in a setting where each node $v$ is required to receive only one \emph{large} $\Theta(n)$-bit message $M_{u,v}$ from a \emph{single} source node $u=g(v)$ rather than from $n$ distinct sources (which is the typical situation when simulating a Congested Clique round). Let $g: V \to V$ be some bijection function where $g(u)\neq g(v)$ for every $u,v$. We observe a simple $2$-round solution for this relaxed task. 
Assume, w.l.o.g., that each $C(M_{u,v})$ has exactly $n$ bits. In such a case, each $g(v)$ sends $C(M_{g(v),v})$ in one round to the $n$ nodes, by sending the $i^{th}$ bit in $C(M_{g(v),v})$ to the $i^{th}$ node $v_i \in V$. In the second round, each $v_i$ sends its received bit to $v$. Since the adversary corrupts a total of $2\alpha n$ bits in each of these two rounds, by setting $\alpha=\delta_C/4$, $v$ can successfully decode the message and recover $M_{g(v),v}$. We do it in parallel for all nodes $v$. This suggests that the main challenge in our original routing instance is not in the number of bits that a node needs to receive, but rather in the large number of sources that hold this information. As one can only bound the number of corrupted edges incident to a given node (rather than for an arbitrary edge set) in a given round, our strategy is:
\begin{strategy}
    \center Concentrate Information on Few Key Nodes $K(u)$ for Each Node $u$.
\end{strategy}
\smallskip
Since there are a total of $n^2$ bits to be delivered through the network, this poses a great challenge and calls for various compression schemes. We start with formulating our strategy by characterizing the routing instances that are solvable in $O(1)$-rounds in the presence of $\alpha$-ABD adversary. This routing formulation will be our replacement to the well known Lenzen's routing \cite{L13} that has been heavily used in fault-free Congested Clique algorithms.

\smallskip
\noindent \textbf{The Super-Message Routing Procedure.} We assume that each node is a source and target of $k=O(1/\alpha)$ super-messages each of $O(\alpha n)$ bits. Let $\{m_j(u)\}_{j \in [k]}$ be the input messages of $u$ and let $\Target(m_j(u))\subseteq V$ be the targets of the message $m_j(u)$.\footnote{This stands in contrast to the Lenzen's routing formulation where each message has a single target.} We require that these target sets will be known to all nodes. In the end of this part, we compare our routing to the fault-free routing formulation by \cite{L13} and its generalizations by \cite{CFG20}.  We start again by encoding each message into $C(m_j(u))$ where $C$ is an ECC with constant rate and distance, e.g., the Justesen code from \Cref{lem:justesen_code}. Assume, w.l.o.g., that each encoded message $C(m_j(u))$ has $L \leq n$ bits. The idea is to locally define at each node $u$ a subset of ordered $L$-nodes $A_{(u,j)}$. The algorithm then has two rounds. In this first, $u$ distributes the $C(m_j(u))$ to the nodes in $A_{(u,j)}$ (by sending the $i^{th}$ bit $C(m_j(u))$ to the $i^{th}$ node in $A_{(u,j)}$). This is done in parallel for all $kn$ messages. In the second step, the nodes in $A_{(u,j)}$ send their received bits to each node $v$ in $\Target(m_{u,j})$. The main challenge is in defining the $A_{(u,j)}$ sets in a balanced way so that each node $u$ sends at most $1$ bit to any other node. A trivial approach is defining $A_{(u,j)}$ to be an arbitrary set of $L$ nodes, but this will lead to a round complexity of $O(1/\alpha)$ by the approach taken in the previous paragraph. As some of our algorithms uses non-constant values of $\alpha$ (in particular, our main algorithm of \Cref{thm:informal_adaptive_randomized}), our goal is to implement the routing in constant number of rounds (independent of $\alpha$).

Towards that goal we use the notion of cover-free sets (e.g. \cite{KRS99}), that have been used frequently in the context of distributed graph coloring algorithms \cite{L92}. Informally, an $(r,\delta)$-cover free family is a family $\mathcal{A}$ of sets where the union of every $r$ sets $A_1,\dots,A_r \in \mathcal{A}$ covers at most a $\delta$ fraction of the elements of any other set $A_0 \in \mathcal{A}$.
The existing constructions of $(r,\delta)$-cover free set do not provide the desired tradeoff between $r,\delta$. We therefore provide a construction\footnote{Our construction is a simple adaptation of a randomized construction in \cite{KRS99} for $(r,\delta)$-cover free sets. In addition to adapting it to the new variant, we also derandomize the construction using a deterministic LLL algorithm.} for a weaker variant of $(r,\delta)$-cover free sets that exploits the fact that the number of constraints in our case is bounded by $2kn$, corresponding to the total number of input and output messages of each node. At the end, our computed $kn$ sets $A_{(u,j)}$ for $j \in [k]$ and $u \in V$ might still admit a large overlap, and hence we should use them with care. 
Specifically, our algorithm explicitly avoids congestion (i.e., omit the overlap) by ensuring that $u$ sends its bit to some $w \in A_{(u,j)}$ only if $w$ does not appears in any other $A_{(u,j')}$ for $j' \neq j$. (Similar constraints are employed on the second step where the received bits are sent to their designated targets). We use the properties of the cover-free sets to show that for every codeword $C(m_j(u))$, there are sufficiently many nodes in $A_{(u,j)}$ that correctly receive their designated bit in $C(m_j(u))$, and that every $v \in \Target(m_j(u))$ receive sufficiently many of these correct bits from $A_{(u,j)}$.

\emph{Comparison to Lenzen's Routing and its Generalizations.} In the classical Lenzen's routing scheme \cite{L13} every node has a total of $nB$ messages where each $B$-message has one designated target. The routing can be implemented in $O(1)$ many $B$-Congest Clique rounds provided that each node is the source and target of at most $n$ messages. A generalization proposed by Censor-Hillel et al.~\cite{CFG20} allows each $B$-bit message to have multiple targets provided that each node is the source and target of at most $O(n)$ messages, and assuming the source and targets of all messages are known to all (See \cite{CFG20} Lemma 9). Our routing scheme can be seen as the analogue of the above to the $\alpha$-BD adversarial setting, with some restrictions compared to \cite{CFG20}: for $k = O(1/\alpha)$, each node can be the source and target of at most $k$ many $O(n/k)$-bit messages. Additionally, the source and targets are known to all nodes, unlike \cite{L13}.

\smallskip
\noindent \textbf{Towards an Adaptive $O(1)$-Round Randomized Algorithm with $\alpha=\widehat{O}(1)$.} Our key result is a $O(1)$-round algorithm for the $\AllToAllComm$ problem in the presence of $\alpha$-ABD setting with $\alpha=1/n^{o(1)}$. Before explaining our approach, we revisit two relaxed settings for the problem and show why their solutions cannot be extended to our setting. The first, handled by the prior work of \cite{FP23} supports a considerably smaller number of only $\widetilde{\Theta}(n)$ edge corruptions, compared to the almost quadratic number of edges in our setting. As we show, the approach of \cite{FP23} fails already for $\alpha=1/n$, i.e., where the adversary controls a matching (one distinct edge per node). 
Moreover, the algorithm of \cite{FP23} runs in $\poly\log n$ rounds while we aim for constant round complexity. The second setting that we consider a warm-up in this paper, supports $\alpha=\Theta(1)$ but only works for non-adaptive adversaries. 

\smallskip
\noindent \textbf{Prior Work \cite{FP23}: Adaptive Adversary with Total $\Theta(n)$ Corruptions and $\poly\log n$ Rounds.}
The most relevant prior work is that of \cite{FP23}.\footnote{We remark that their \cite{FP23}'s techniques are aimed towards a wider class of topologies. Here we present their approach only through the lens of the Congested Clique.} Their algorithm works as follows: First, each node $u$ sends the message $m_{u,v}$ to $v$, which receives a potentially corrupted message $\widetilde{m}_{u,v}$. Their goal becomes to detect all messages such that $\widetilde{m}_{u,v} \neq m_{u,v}$ and correct them. To do so, they partition the $n$-clique into $n$ (nearly disjoint\footnote{Each edge appears in exactly two trees.}) spanning trees, where the $i^{th}$ tree is taken to be all edges incident to $v_i$. Since the adversary controls only $cn$ edges for some small constant $c$, a majority of these $n$ spanning trees are reliable in a given round.
They then attempt to aggregate information on each tree to discover all messages $\widetilde{m}_{u,v} \neq m_{u,v}$, using sketching based techniques, and specifically sparse recovery sketches, on which we elaborate later on. 
Using these sketches, each node $v$ is required to receive $\poly\log n$ bits from which it can deduce its corrupted messages. 

While the sketching technique of \cite{FP23} is very useful in our $\alpha$-ABD setting, our general approach is quite different due to the following observation. For \cite{FP23} it was critical to have a majority of reliable spanning trees in a given round. This no longer holds in our setting, even for $\alpha=1/n$ as the adversary can still corrupt an edge in each tree \emph{simultaneously} in a given round. We therefore must take a different approach for aggregating the sketch information, as explained later on.

\smallskip
\noindent \textbf{This Work (Warm-Up): Non-Adaptive Adversary with $\alpha=\Theta(1)$ and $O(1)$ Rounds.}
This algorithm only works for $B$-Congested Clique for $B=\Theta(\log n)$, which is the standard bandwidth assumption for this model. 
The algorithm is simple, and exemplifies the gap in adversary strength between adaptive and non-adaptive settings.

Consider the following simple sub-procedure: For each $v \in V$, send \emph{all} messages $\{m_{u,v}\}_{u \in V}$ addressed to $v$ to a single random node $v'$. Then, using the resilient super-message routing procedure, send the $n$ messages received by $v'$ to $v$. Since the adversary is oblivious to randomness, one can easily show that each bit message $m \in \{m_{u,v}\}_{u \in V}$ arrives uncorrupted to $v$ with probability at least $1-\alpha$.  Given this sub-procedure, we show how to solve the $\AllToAllComm$ problem. We assume that $B=\Theta(\log{n})$, therefore we can (a) encode each message using an ECC with constant rate (of Lemma \ref{lem:justesen_code}), and (b) split each resulting codeword into $\Theta(\log{n})$ parts, each of one bit. We then apply the sub-procedure above for each part independently. Consequently, each received codeword $\widetilde{C}(m_{u,v})$ for every $u \in V$, has $O(\alpha)$ fraction of corruptions w.h.p., and thus $v$ can recover the messages $\{m_{u,v}\}_{v \in V}$. 

We remark that this approach fails completely for an adaptive adversary: if the adversary is aware of all randomness generated, we cannot elude it by sending the information through a random edge. We also remark that randomization is paramount: for deterministic algorithms, it is easy to see that non-adaptive and adaptive adversaries are equivalent (since in both settings the adversary knows the initial state, then in a deterministic setting it can deduce the states of all nodes at all times).

\smallskip
\noindent \textbf{Finally: Adaptive Adversary with $\alpha=\widehat{O}(1)$ and $O(1)$ Rounds.}
Let $C$ be a non-adaptive LDC code with constant distance and rate and with query complexity of $q=n^{o(1)}$ as obtained by \Cref{lem:LDC}
To ease presentation, we first show an algorithm that runs in $O(q)$ rounds. Then, we integrate also the tool of sparse recovery sketches to reduce the round complexity down to $O(1)$.

\smallskip
\noindent \textbf{Take I: $n^{o(1)}$ Rounds.} 
Initially, each node $u$ has a $1$-bit message $m_{u,v}$ for each node $v$ in the network. The node $u$ first concatenates all its $\{m_{u,v}\}_{v \in V}$ into an $n$-bit message $M(u,V)$ where $M(u,V)[i] = m_{u,v_i}$, and encodes it into a codeword $C(M(u,V))$ with $O(n)$ bits. 
$u$ partitions the codeword $C(M(u,V))$ into $n$ parts each of size $O(1)$ bits, and sends the $j^{th}$ part to each node $v_j$ in the network (where some parts may be corrupted). This can be done in $O(1)$ rounds. 
The goal of each node $v_i$ is to recover $M(u,V)[i] = m_{u,v_i}$ for all $u \in V$. 

In order for $v_i$ to recover $M(u,V)[i]$ for some $u$, it should ``query'' $q = n^{o(1)}$ random nodes by using our resilient routing scheme, obtain the $q$ nodes' parts of $C(M(u,V))$, and recover $M(u,V)[i]$. To be able to recover $M(u,V)[i]$ for all $u \in V$ \emph{simultaneously}, we use the non-adaptivity of the LDC $C$. The latter property implies that the identity of the $q$ random nodes to be queried depend only on the identity of $v_i$. Hence, $v_i$ needs to query the same $q$ nodes for all $u \in V$. This exactly fits our strategy of concentrating the information of $v$ on a small set of $q$ nodes.  We can use our super-messages routing scheme to send all codeword parts of the $q$ nodes to $v_i$, allowing it to obtain all parts and recover $\{m_{u,v_i}\}_{u \in V}$. In this process, each node learns $O(nq)$ bits, and indeed it would take $O(q) = n^{o(1)}$ rounds to implement.

\smallskip
\noindent \textbf{Take II: $O(1)$ Rounds.} In the previous algorithm each node was required to learn $n$ bits of information, and using LDCs, for each bit it had to query $q$ bits of the encoded information, leading to a total of $qn$ bits which requires $q$ rounds. To improve the round complexity, our goal is to reduce the total number of bits that a node should recover from LDC codewords to $\widetilde{O}(\alpha n)=o(n)$. Since for each such bit, $u$ should query $q$ positions in the codeword, it should receive a total of $\widetilde{O}(q\alpha n)=O(n)$ bits, where the last equality follows by setting $\alpha=1/(q \polylog n)$. 

For that purpose, we use $k$-sparse recovery sketches. The protocol starts by letting all nodes exchange their $m_{u,v}$ messages. Letting $\widetilde{m}_{u,v}$ be the received message at each node $v$, we note that there are at most $\alpha n$ corrupted messages with $\widetilde{m}_{u,v}\neq m_{u,v}$ per node $v$. The goal of each node is to learn which of its received messages are corrupted along with their correction.

We take the following approach. Let $P = \{P_1,\dots,P_{\alpha n}\}$ be a random partition of the nodes where each set $P_j$ is of size $O(1/\alpha)$. This random partition is computed after the execution of the first round, and therefore by Chernoff it holds that the number of corrupted messages that $v$ received from the nodes in $P_j$ is at most $O(\log n)$, w.h.p. I.e., $|\{\widetilde{m}_{u,v} ~\mid~ u \in P_j, \widetilde{m}_{u,v}\neq m_{u,v}\}|=O(\log n)$. 

This allows us to use $k$-sparse recovery sketch of \Cref{lem:k_sparse_recover} with parameter of $k=O(\log n)$ bits (see Lemma \ref{lem:sparse}). For each $P_j \in P$ and $v \in V$, the nodes in $P_j$ compute a $k$-sparse recovery sketch $\Sk(P_j,v)$ obtained by adding to the sketch the (uncorrupted) $M(P_j,\{v\})$ messages. Each sketch $\Sk(P_j,v)$ has a total of $O(k\log^2 n)=\widetilde{O}(1)$ bits. By using carefully the super-message routing procedure, the algorithm 
encodes the $\{\Sk(P_j,v)\}_{v \in V, j \in [\alpha n]}$ sketches using the LDC code and this encoded information is held in a distributed manner by all nodes in $V$. To recover its designated $\{m_{u,v}\}_{u \in V}$ messages, it is sufficient for each $v$ to learn its $\alpha n$ sketches $\{\Sk(P_j,v)\}_{j \in [\alpha n]}$. For each bit in that sketch it needs to query $q$ nodes, hence $v$ is required to obtain $\widetilde{O}(q \alpha n)=o(n)$ bits from which it can deduce the sketch $\{\Sk(P_j,v)\}_{j \in [\alpha n]}$ and hence its messages $\{m_{u,v}\}_{u \in V}$. We implement each step by a constant number of super-message routing instances, which can be completed in a total of $O(1)$ rounds.

\smallskip
\noindent \textbf{Deterministic Algorithms.}
So far, we presented randomized solutions. We next describe two deterministic algorithms for different regimes of $\alpha$. 

\smallskip 
\noindent \emph{(I) $\alpha=O(1)$ with $O(\log{n})$ Rounds.} 
For $\alpha \leq c$ for some small constant $c \in (0,1)$, we show a deterministic $O(\log{n})$ round algorithm for $\AllToAllComm$. By Lemma \ref{lem:change_n_value}, we may assume that $n$ is a power of two. Also let the nodes' IDs be numbered $\{0,\dots,n-1\}$.

On a high level, our protocol has a relatively simple communication pattern: We have $\log{n}$ iterations. At each iteration $i$, we match each node $u$ with a node $u'$, where $u$ and $u'$ are matched if and only if their ids agree on all bits except $i^{th}$ bit. Each matched pair exchanges messages with each other, and no other node, using our resilient routing scheme.

Specifically, at the beginning of iteration $i$, each node $u$ holds a set of messages $M_i(u)$, where initially, $M_1(u) = M(\{u\},V)$. It sorts the messages by target-ID in descending order, and defines $M^-_i(u)$ and $M^+_i(u)$ as the lower and upper half of $M_i(u)$ according to the ordering, respectively. When $u$ and $u'$ exchange messages, one node receives $M^-_i(u),M^-_i(u')$, and one node receives $M^+_i(u),M^+_i(u')$ (according to which has the higher ID). If node $u$ receives $M^-_i(u),M^-_i(u')$, it sets $M_{i+1}(u) = M^-_i(u) \cup M^-_i(u')$ and otherwise, $M_{i+1}(u) = M^+_i(u) \cup M^+_i(u')$.

We are able to characterize the message set $M_i(u)$ that node $u$ holds at the beginning of iteration $i$ in the following manner. 
let $P(u,i)$ is the set of nodes whose IDs agree with $\ID(u)$ on the first $i$ bits, and $S(u,i)$ the set of nodes that agree on the last $\log n-i+1$ bits. We prove by induction on $i$ that $M_i(u)=M(S(u,i),P(u,i))$. Hence, at the end of the last iteration, each node $u$ holds $M_{\log{n}+1}(u) = M(V,\{u\})$, and may output this set.  

\smallskip
\noindent\emph{(II) $\alpha = \Theta(1/\sqrt{n})$ with $O(1)$ Rounds.} 
The final algorithm is a simple application of the resilient super-message routing scheme. A very similar construction is provided in \cite{ABEGH19} (Theorem 1, therein) for the stochastic interactive coding setting, which is quite different than the Byznatine setting considered in this work.  For $\alpha = O(1/\sqrt{n})$, the scheme allows each node to safely send and receive $\Theta(\sqrt{n})$ super-messages of size $\Theta(\sqrt{n})$ each. For simplicity, assume $\sqrt{n}$ is an integer. We split $V$ into $\sqrt{n}$ equal sized sets $S_1,\dots,S_{\sqrt{n}}$. The algorithm has two steps: in the first step, we make sure that each $S_i$ collectively holds $M(S_i,V)$ (i.e., all $m_{u,v}$ messages from $u \in S_i$ to $v \in V$), by having a node $S_i[j]$ learn the message set $M(S_i,S_j)$ using the routing scheme. In the second step, we have each node $v$ learns $M(V,\{v\})$ by learning each of $M(S_1,\{v\}),\ldots,M(S_{\sqrt{n}},\{v\})$ (i.e., $\sqrt{n}$ super-messages each of $\sqrt{n}$ bits) by another application of the routing scheme.

    \section{Resilient Super-Message Routing}
\label{sec:key_routing}

For two given integer parameters $k,\lambda$, an instance of the $\SuperMessagesRouting$ is described as follows: Each node $u$ is given as input $k$ messages $\Minput(u) = \{m_1(u),\dots,m_k(u)\}$, each of at most $\lambda$ bits, where $\lambda$ is possibly much larger than the bandwidth parameter $B$. Hereafter, a $\lambda$-bit message is referred to as a \emph{super-message}. Each super-message $m_j(u)$ is required to be sent to a list of targets $\Target(m_j(u)) \subseteq V$ such that each node is in the target list of at most $k$ super-messages. Importantly, all the nodes know the target lists of all the $kn$ super-messages, given by $\{\Target(m_j(u))\}_{u\in V, j \in [k]}$. Specifically, if we identify each super-message by the tuple $(u,j)$, each node knows the target list of each super-message $(u,j)$. The goal is for each node $u$ to learn the $k$ super-messages addressed to it, i.e., to output
 $$\Moutput(u) = \{m_j(v) \stt u \in \Target(m_j(v)), j \in [k], v \in V\},$$
and more specifically, for each $m_j(v) \in \Moutput(u)$, it outputs  $m^u_j(v) = m_j(v)$.

Our key end result is a deterministic algorithm for solving the $\SuperMessagesRouting$ problem with input parameter $k$ and $\lambda$ in the presence of an $\alpha$-ABD adversary for $\alpha\leq 1/(8 \cdot 10^4)$. 

\begin{theorem}
\label{thm:generalized_routing_main}
Let $\alpha \leq 1/(8 \cdot 10^4)$, and let $k = O(\min(1/\alpha,n/\log n))$ be an integer. For any bandwidth $B$, there is a deterministic protocol $\ExtMatchingTransmission$ that can solve the $\SuperMessagesRouting$ problem with parameters $k$ and $\lambda$ in the $\alpha$-ABD setting within $O(k\lambda/(Bn))$ rounds. 
\end{theorem}

In particular, for $\lambda = \Theta(nB/k)$, the routing procedure above terminates in $O(1)$ rounds. To prove the theorem, we prove the following lemma, from which the theorem follows as a corollary.

\begin{lemma}
\label{lem:ext_matching_transmission}
Let $c,\hat{c} > 0$ be some sufficiently small universal constants. Assume $B = 1$. Let $\alpha \leq 1/(8 \cdot 10^4)$, and let $k \leq \min(\lfloor 1/(8 \cdot 10^4 \alpha) \rfloor,\hat{c}n/\log{n})$ be an integer. There is a deterministic $O(1)$-round protocol $\ExtMatchingTransmission$ for solving the $\SuperMessagesRouting$ problem with parameters $k$ and $\lambda=cn/k$ in the $\alpha$-ABD setting.
\end{lemma}

On a high level, our routing procedure has a two-round structure. Initially, each node $u$ encodes each of its input messages $m_j(u)$ using an error correction code $C$. In the first round, for every $j \in \{1,\ldots, k\}$, the node $u$ sends an $L$-bit encoded message $C(m_j(u))$ to a subset $A_{(u,j)}\subseteq V$ of $L$ nodes, ordered by the vertex-IDs, where the $i^{th}$ node of $A_{(u,j)}$ receives the $i^{th}$ bit of $C(m_j(u))$. In the second round, the nodes $A_{(u,j)}$ send their received bit to each $v \in \Target(m_j(u))$. The main challenge is in defining the set $A_{(u,j)}$. Towards that goal we use the notion of $(r,\delta)$-Cover-Free Sets which can be viewed as a stronger variant of the well-known $r$-cover-free sets, that have been used frequently in the context of distributed graph coloring \cite{L92}. 

\subsection{Cover-Free Sets}

We start by providing the definition of $(r,\delta)$-cover-free sets. 

\begin{definition}[$(r,\delta)$-Cover-Free Sets]
Let $N,r$ be integers, and let $0 < \delta < 1$. Let  $\mathcal{A} = \{A_1,A_2,\dots,A_m\} \subseteq 2^{[N]}$ for some $m \geq r+1$. Then $\mathcal{A}$ is called an $(r,\delta)$-covering set if for any $(r+1)$ distinct sets $A_{i_0},\dots,A_{i_{r}} \in \mathcal{A}$, it holds that $|A_{i_0} \setminus \bigcup_{j=1}^{r} A_{i_j}| \geq (1-\delta)|A_{i_0}|$. 
\end{definition}

The notion of $(r,\delta)$-cover-free sets is closely related to the standard notion of $r$-cover-free sets, where it is only required that no set in the family be fully covered by $r$ other sets.

The seminal work of \cite{KRS99} introduced $(r,\delta)$-cover-free sets and showed several constructions with various parameters. Unfortunately, none of these constructions obtain parameters sufficient for our purposes. We therefore define a slightly weakened variant of $(r,\delta)$-cover-free set, and show a construction of such a family, inspired by one of the constructions of \cite{KRS99}.

\begin{definition}[$(r,\delta)$-Cover-Free Sets w.r.t.~$H$]
    Let $N,r$ be integers, and let $0 < \delta < 1$. A set $\mathcal{A} = \{A_1,\dots,A_m\} \subseteq 2^{[N]}$ is called an $(r,\delta)$-covering set w.r.t. $H$ for a set $H \subseteq \mathcal[m]^{r+1}$, if for any $i_0 \in [m]$ and $\{i_0,\dots,i_{r}\} \in H$ it holds that $|A_{i_0} \setminus \cup_{j = 1}^{r} A_{i_j}| \geq (1-\delta)|A_{i_0}|$.
\end{definition}

For a family of sets $\mathcal{A} = \{A_1,\dots,A_m\}$, we say that the set size of $\cA$ is $x$ if $|A_i| = x$ for all $i \in [m]$. We remark that any family $\mathcal{A}$ is trivially a $(0,\delta)$-cover-free set for any $0 < \delta < 1$.

\begin{lemma}
    \label{lem:cover_free_proof}
    Let $N$ be a positive integer, and let $0 < \delta < 1$ be a constant. Let $c_1 \in (0,1),c_2 > 0$ be constants such that $c_1$ is sufficiently small compared to $c_2$. Let $r \geq 0$ be an integer such that $r+1 \leq c_1 \delta N/\log{N}$. Then for any integer $m \geq r+1$ and any $H \subseteq [m]^{r+1}$ of size $|H| = O(N^{c_2})$, there exists an $(r,\delta)$-cover-free set $\mathcal{A} \subseteq 2^{[N]}$ w.r.t. $H$ of size $|\mathcal{A}| = m$, and of set size $L = \lfloor \delta N/(4r+4) \rfloor$. Moreover, such a family can be computed sequentially deterministically in polynomial time.
\end{lemma}
In this section, we give a randomized construction of the claim above, and show an efficient derandomization of it in Appendix~\ref{sec:cover_free_lll_derandomization}. We remark that this simplified randomized construction is sufficient to imply \Cref{thm:generalized_routing_main}, though without the efficient local time computation guarantee (via a trivial brute-force derandomization).

\begin{proof}[Proof of randomized variant of \Cref{lem:cover_free_proof}]
 The construction is similar to the randomized construction of \cite{KRS99}. Let $\gamma = \delta/4$. Therefore, $L = \lfloor N\gamma/(r+1) \rfloor$. Partition $[N]$ into consecutive sets $S_1,\dots,S_L$ of size $\lfloor\frac{r+1}{\gamma} \rfloor$ each (ignoring the remaining elements). Each set $A_1,\dots,A_m$ is an i.i.d. random set which contains exactly one uniformly random element from each of $S_j$, i.e. each element is chosen with probability $1/\lfloor \frac{r+1}{\gamma} \rfloor \leq 2/(\frac{r+1}{\gamma}) = \frac{2\gamma}{r+1}$, where the first inequality follows due to the fact that for any $x \geq 1$, $x \leq 2\lfloor x \rfloor$. Indeed, the size of each set is $|A_i| = L$ for all $i \in [m]$.
    
For any fixed $r+1$ distinct sets $A_{i_0},\dots,A_{i_{r}}$, it holds by union bound that the unique element $a \in A_{i_0} \cap S_j$ is contained in any of $A_{i_1},\dots,A_{i_r}$ with probability at most $r \frac{2\gamma}{r+1} \leq 2\gamma = \delta/2$. Let $I_j(i_0,\dots,i_{r})$ be an indicator variable that this occurs. Then for $X(i_0,\dots,i_{r}) = \sum_{j=1}^{L} I_{j}(i_0,\dots,i_{r})$, it holds that $E(X(i_0,\dots,i_{r})) \leq (\delta/2) \cdot L $. We note that $I_1(i_0,\dots,i_{r}),\dots,I_{L}(i_0,\dots,i_{r})$ are independent random variables. By choice of $c_1$, we get that $L = \lfloor N\gamma/(r+1) \rfloor \geq c'\log{N}$, for some sufficiently large $c' > 0$, and in particular we assume $c' \geq 24+12c_2$. Therefore by Chernoff, 
$$\Pr(|A_{i_0} \setminus \bigcup_{j=1}^{r} A_{i_j}| < (1-\delta)|A_{i_0}|) = \Pr\left(X > \delta|A_{i_0}|\right) = \Pr\left(X > \delta L\right)  \leq e^{-c'\log{N}/12} \leq e^{-(2+c_2)\log{N}}\leq \frac{1}{N^2|H|}.$$ By union bound over all choices of $i_0 \in [m]$ and $\{i_0,\dots,i_{r}\} \in H$ (there are $(r+1)|H| \leq N|H|$ many such choices - one for each tuple in $H$ and choice of $i_0$ from the tuple), it holds that for any $i_0 \in [m],\{i_0,\dots,i_{r}\} \in H$, the sets $A_{i_0},\dots,A_{i_{r}}$ satisfy this condition with probability $1-\frac{1}{N}$.
\end{proof}

\subsection{Deterministic Algorithm for the Super-Message Routing}
Let $\IN(u)=\{(u,1), \ldots, (u,k)\}$ be the indices of the super-messages given as input to $u$, and $\OUT(u)=\{(v,j) \stt u \in \Target(m_j(v))\}$ be the indices of the super-messages that $u$ should receive. We then have that $|\IN(u)|,|\OUT(u)|\leq k$, and that all nodes know the $2n$ sets $\{\IN(u)\}_{u \in V} \cup \{\OUT(u)\}_{u \in V}$. 

\smallskip
\noindent \textbf{Round 0 (Local Computation).} Let $0 <\delta < 1$ be a constant sufficiently close to zero, that we fix in the analysis. The algorithm starts by a local computation of a $(k-1,\delta)$-cover-free set $\cA$ w.r.t. a collection $H$ of $2kn$ sets
each of size $\lfloor \frac{\delta n}{4k} \rfloor$, given by: 
\begin{equation}\label{eq:H}
H= \{\IN(u)\}_{u \in V} \cup \{\OUT(u)\}_{u \in V}.
\end{equation}
To do so, each node uses the following construction follows from \Cref{lem:cover_free_proof}.

\begin{lemma}
    \label{lem:cover_free_routing}
    Given $0 < \delta < 1$ and set $H$, every node can locally and deterministically compute a $(k-1,\delta)$-cover-free set $\cA = \{A_{(v,j)}\}_{v \in V,j \in [k]}$ w.r.t $H$ with set size $L=\lfloor \delta n/(4k) \rfloor$ and each $A_{(v,j)} \subseteq V$.  
\end{lemma}
\begin{proof}
Recall that we assume $k \leq \hat{c}n/\log{n}$ for some small constant $\hat{c} > 0$. Therefore, by \Cref{lem:cover_free_proof} with parameters $N=n,r=k-1, m = 2nk$, there exists a  $(k-1,\delta)$-cover-free set $\cA = \{A_{(v,j)}\}_{v \in V,j \in [k]}$ with set size $\lfloor \delta n/(4k) \rfloor$. While according to \Cref{lem:cover_free_proof}, each set $A_i \in \cA$ is a subset of $[n]$, we instead assume that $A_i \subseteq V$ using the natural bijection between $[n]$ and $V$ that maps each $i \in [n]$ to $v_i \in V$.
\end{proof}

From this point on, we assume that all nodes agree on a  $(k-1,\delta)$-cover-free set $\cA$ w.r.t $H$. Let $C$ be the Justesen code (\Cref{lem:justesen_code}). 
Each node $v$ computes the codewords $C(m_1(v)),\dots,C(m_k(v))$. Recall that we assume that each super-message has $cn/k$ bits for sufficiently small $c > 0$. In particular we assume $c \leq 200c'$ where $c'$ is the constant such that $\cA$ has set size $c'n/k$.  Therefore, we may assume that each $C(m_j(v))$ has $|A_{(v,j)}|$ bits. 
Each vertex set $A_{(v,j)}$ is considered to be lexicographical ordered (e.g., in increasing order of the nodes IDs). Define:
$$\inload(v,w)=|\{ j \in [k] \stt w \in A_{(v,j)}\}| \mbox{~and~} \outload(w,v)=|\{ (u,j)\in \OUT(v) ~\mid~ w \in A_{(u,j)}\}|~.$$

That is, $\inload(v,w)$ bounds that total number of bits that $v$ needs to send $w$ when sending its input encoded messages to the $A_{(v,j)}$ sets. Similarly, $\outload(w,v)$ bounds that total number of bits that $v$ needs to receive from $w$ when receiving its encoded target messages.

Since the nodes are provided with the target sets $\{\Target(m_j(u))\}_{u \in V,j \in [k]}$ and with the cover-free set $\mathcal{A}$, each node $u$ can locally compute $\inload(v,w), \outload(w,v)$ for every $w, v \in V$.

\smallskip
\noindent \textbf{Round 1:} Each vertex $u$ applies the following procedure for each $j \in [k]$. Let $A_{(u,j)}=\{w_1,\ldots, w_L\}$. For every $\ell\in [L]$, $u$ sends the $\ell^{th}$ bit in $C(m_j(u))$ to node $w_\ell$ \textbf{only if} $\inload(u,w_\ell)=1$. For every $u,w \in V$, let $b_1(u,w)$ be the bit received at $w$ from node $u$ at the end of this round; if no bit is received, then set $b_1(u,w)=0$.

\smallskip
\noindent \textbf{Round 2:} Each vertex $w$ applies the following procedure for each of its (at most $n$) received bits 
 $b_1(u,w)$, provided that $\inload(u,w)=1$. Let $j \in [k]$ be the unique index such that $w \in A_{(u,j)}$ (since $\inload(u,w)=1$, such an index $j$ exists). Then $w$ sends the bit $b_1(u,w)$ to each vertex $v$ such that (i) $v \in \Target(m_j(u))$ and (ii) $|\outload(w,v)|=1$. For every $w, v \in V$, let $b_2(w,v)$ be the bit received at $v$ from node $w$. If no bit is received from $w$ at $v$, then set $b_2(w,v)=0$. 

\smallskip
\noindent \textbf{Output:} In this step, each vertex $v \in V$ locally computes its output to each of its $k$ super-messages, i.e. a set
$\{m^v_j(u) \mid (u,j)\in \OUT(v)\}$, where $m^v_j(u)$ corresponds to the message $m_j(u)$. To output the super-message $m^v_j(u)$, node $v$ performs the following. Let $A_{(u,j)}=\{w_1,\ldots, w_L\}$. Set $\widetilde{C}_v(m_j(u)) = (b_2(w_1,v),\dots,b_2(w_L,v))$ be the bits $v$ receive corresponding to the message $C(m_j(u))$. Set $m^v_j(u)=\Decode(\widetilde{C}_v(m_j(u)))$. Each node $v$ outputs the set $\Moutput(v) = \{m^v_j(u) \stt (u,j)\in \OUT(v)\}$.

\smallskip\noindent \textbf{Analysis.} For the remainder of the section, we set $\delta = 1/50$. Recall that $\delta_C$ denotes the distance of the error correcting code of \Cref{lem:justesen_code}.

\begin{lemma}
\label{lem:secure_routing_parameter_calculation}
The following inequalities hold:
(a) $(16/\delta)\alpha k + 2\delta < \delta_C/2$; (b) For any $A \in \cA$, $\alpha n \leq (8/\delta) \alpha k|A|$.
\end{lemma}
\begin{proof}
(a). Recall that $k \leq \lfloor 1/(8 \cdot 10^4 \alpha) \rfloor$, and $\delta = 1/50$. Therefore,
    $$(16/\delta)\alpha k+2\delta  \leq (1/100)+(1/25) = 1/20 < \delta_C/2.$$ 
(b). By Lemma \ref{lem:cover_free_routing}, $|A| = \lfloor \delta n/(4k) \rfloor$. It holds that $\delta n/(4k) \leq 2\lfloor \delta n/(4k) \rfloor = 2|A|$. By rearranging, $\alpha n \leq (8/\delta)\alpha k|A|$.
\end{proof}

\begin{lemma}
    \label{lem:secure_corrupt_bound}
    For any $v \in V$ and $(u,j) \in \OUT(v)$, it holds that \\$\Hamm(\widetilde{C}_v(m_j(u)),C(m_j(u))) < \delta_C \cdot |C(m_j(u))|/2$. 
\end{lemma}
\begin{proof}
    Let $v \in V$, and $(u,j) \in \OUT(v)$. Let $w_1,\dots,w_{L}$ be the elements of $A_{(u,j)}$ ordered by lexicographical order. Set $M_{\mathrm{mid}} = b_1(u,w_1),\dots,b_1(u,w_L)$, i.e. the string the nodes of $A_{(u,j)}$ receives after the first round.
    
    First, we bound $\Hamm(M_{\mathrm{mid}},C(m_j(u)))$. We notice that $M_{\mathrm{mid}}[i] \neq C(m_j(u))[i]$ if at least one of two cases occur: (a) $\inload(u,w_i) > 1$; (b) the adversary corrupted the bit $b_1(u,w_i)$. We bound the number of indices that each case can occur, separately: 
    \begin{enumerate}[(a)]
        \item Recall that $\{(u,1),\dots,(u,k)\} \in H$. Therefore, only $\delta |A_{(u,j)}|$ elements of $A_{(u,j)}$ appear in any of $A_{(u,1)},\dots,A_{(u,j-1)},A_{(u,j+1)},A_{(u,k)}$.
        Hence, case (a) occurs in at most $\delta |A_{(u,j)}| = \delta |C(m_j(u)|$ indices. 
        \item We note that $\deg_{F_1}(u) \leq \alpha n \leq (8/\delta)\alpha k |A_{(u,j)}| = (8/\delta)\alpha k |C(m_j(u)|$, where the second inequality is shown in \Cref{lem:secure_routing_parameter_calculation}(b). Therefore, the adversary may corrupt at most $(8/\delta)\alpha k |C(m_j(u)|$ many indices of $M_{\mathrm{mid}}$ from the original value of $C(m_j(u))$.
    \end{enumerate}
    
    Overall, $$\Hamm(M_{\mathrm{mid}},C(m_j(u))) \leq (\delta+(8/\delta)\alpha k)|C(m_{j}(u))|.$$ Next, we bound $\Hamm(M_{\mathrm{mid}},\widetilde{C}_v(m_j(u)))$ in a very similar manner.   Notice that $M_{\mathrm{mid}}[i] \neq \widetilde{C}_v(m_j(u))[i]$ only if one of two cases occur: (a) $\outload(w_i,v) > 1$, or (b) if the adversary corrupted the bit $b_2(w_i,v)$ sent from $w_i$ to $v$. We again bound each case separately:  
    \begin{enumerate}[(a)]
        \item  For case (a), by Eq. (\ref{eq:H}), $\OUT(v) \in H$. Therefore, at most $\delta |A_{(u,j)}| = \delta|C(m_j(u)|$ elements $w \in A_{(u,j)}$ have $\outload(w,v) > 1$. 
        \item  Since $\deg_{F_2}(u) \leq \alpha n \leq (8/\delta)\alpha k |C(m_j(u)|$, then the adversary may corrupt at most $(8/\delta)\alpha k |C(m_j(u)|$ many indices of $\widetilde{C}_v(m_j(u))$ from the original value of $M_{\mathrm{mid}}$.
    \end{enumerate}
    Overall, $\Hamm(M_{\mathrm{mid}},\widetilde{C}_v(m_j(u))) \leq (\delta+(8/\delta)\alpha k)|C(m_{j}(u))|$. By the triangle inequality, 
    \begin{flalign*}
        \Hamm(\widetilde{C}_v(m_j(u)),C(m_j(u))) &\leq \Hamm(M_{\mathrm{mid}},C(m_j(u))) + \Hamm(M_{\mathrm{mid}},\widetilde{C}_v(m_j(u))) \\ &\leq (2\delta+(16/\delta)\alpha k)|C(m_j(u))| < \delta_C \cdot |C(m_j(u))|/2,
    \end{flalign*}
    where the last inequality is shown in \Cref{lem:secure_routing_parameter_calculation}(a).
\end{proof}

\begin{lemma}
    \label{lem:correctness_extended}
    For every $v \in V, (u,j) \in \OUT(v)$, it holds that $\Decode(\widetilde{C}_v(m_j(u)))=m_j(u)$.
\end{lemma}
\begin{proof}
    By \Cref{lem:secure_corrupt_bound}, for any  $(u,j) \in \OUT(v)$ it holds that $\Hamm(\widetilde{C}_v(m_j(u)),C(m_{j}(u))) < \delta_C \cdot |C(m_j(u))|/2$. Since the fraction of corrupted bits is less than $\delta_C/2$, the decoding returns $\Decode(\widetilde{C}_v(m_j(u))) = m_{j}(u)$.
\end{proof}

\Cref{lem:ext_matching_transmission} follows immediately from \Cref{lem:correctness_extended}, since each node $v$ outputs $\{m^v_j(u) \stt (u,j)\in \OUT(v)\} = \{m_j(u) \stt (u,j)\in \OUT(v)\}$. Next, we prove \Cref{thm:generalized_routing_main}:

\begin{proof}[Proof of \Cref{thm:generalized_routing_main}]
    We may assume w.l.o.g. that $k \leq \min(\lfloor 1/(8 \cdot 10^4 \cdot \alpha) \rfloor,\hat{c}n/\log{n})$, otherwise we can repeat the procedure of \Cref{lem:ext_matching_transmission} a constant number of times, each time defining a different set of messages. A $\SuperMessagesRouting$ instance with parameters $k$ and $\lambda$ can be partitioned into $O(k\lambda/n)$ instances of $\SuperMessagesRouting$ with $k$ input messages per node, and message size $cn/k$, by splitting each input message into messages of size $cn/k$ bits each. On each such instance, we can apply the procedure of \Cref{lem:ext_matching_transmission} to solve it in $O(1)$ rounds. Moreover, the network has bandwidth $B$, while the protocol of \Cref{lem:ext_matching_transmission} only uses $1$-bit of bandwidth. Hence we may run in parallel $B$ procedures of \Cref{lem:ext_matching_transmission}. We conclude that $\SuperMessagesRouting$ with parameters $k$ and $\lambda$ can be solved in $O(k\lambda/(Bn))$ rounds.
\end{proof}

This completes the proof of \Cref{thm:generalized_routing_main} (i.e., \Cref{thm:informal_routing}). As an immediate corollary, we obtain a procedure that allows one node to broadcast a message of size $O(n)$ to the rest of the network. Formally, in the $\mathrm{Broadcast}$ problem, a special node $v$ initially holds a super-message $m$ of size $O(n)$ bits. The goal is for all nodes to output $m$. We remark that we prove this for $\alpha \leq 1/(8 \cdot 10^4)$, which is sufficient for our purposes, but this constant can be significantly improved for this specific case.

\begin{corollary}
\label{cor:broadcast}
Assume $\alpha \leq 1/(8 \cdot 10^4)$, then the $\mathrm{Broadcast}$ problem can be solved deterministically in $O(1)$ rounds. 
\end{corollary}
\begin{proof}
	 We use $\ExtMatchingTransmission$ where $v$ is the source of a single super message $\Minput(v) = \{m\}$, and for any other node $u$ we set $\Moutput(u) = \{m\}$. \Cref{thm:generalized_routing_main} guarantees that Proc. $\ExtMatchingTransmission$ terminates correctly in $O(1)$ rounds.
\end{proof}
    \section{Randomized Compilers}
\label{sec:randomized_compilers}
In this section we provide randomized compilers against non-adaptive and adaptive adversaries. Specifically, we provide randomized algorithms for the $\AllToAllComm$ problem that captures the single-round simulation in the Congested Clique model. Recall, that in this problem each node $u$ holds $B$-bit messages $m_{u,v}$ for every $v \in V$ and it is required for each $v$ to learn its message set $\{m_{u,v}\}_{u \in V}$. 

\subsection{Handling Non-Adaptive Adversaries} \label{sec:nonadaptive}
In this section, we require bandwidth of $B = \Theta(\log{n})$.

\begin{theorem}
\label{thm:non_rusing_main}
For $\alpha \leq 1/(8 \cdot 10^4)$, there is a randomized protocol $\StaticAllToAll$ for the $\AllToAllComm$ problem in the $\alpha$-NBD setting in $O(1)$ rounds, assuming bandwidth $B = \Theta(\log{n})$.
\end{theorem}

\noindent\textbf{Useful Notation.} For an $L$-bit message $m$ for $L\leq n$ held by a node $u$ and a subset $A \subseteq V$ of $L$ nodes, let $\Send(u,A,m)$ be 
the single-round procedure in which $u$ sends the $i^{th}$ bit in $m$ to the $i^{th}$ node in $A$, where the nodes in $A$ are ordered lexicographically, unless otherwise mentioned. A $L$-bit message $m$ is \emph{distributively stored} by $A$ if the $i^{th}$ bit of $m$ is held by the $i^{th}$ node in $A$. For an $L$-bit message $m$ that is distributively stored by a subset $A \subseteq V$ of $L$ nodes and a vertex $v \in V$, let $\Send(A,v,m)$ be the single-round procedure where the $i^{th}$ node in $A$ sends the $i^{th}$ bit of $m$ to $v$.

\smallskip
\noindent \textbf{Description of Algorithm $\StaticAllToAll$:} Let $C$ be the Justesen code from \Cref{lem:justesen_code}. Assume w.l.o.g. that node $u$ holds $B'$-bit messages $m_{u,v}$ for every $v$, of size $B' = \lceil \log{n} \rceil$ (if our input messages are larger, we can split the messages in a constant number of parts, and apply this routing algorithm a constant number of times to transmit the full messages). Hence, each codeword $C(m_{u,v})$ has size $ B'/\tau_C \leq B$ bits; for simplicity, assume for the remainder of the section that $B = B'/\tau_C$. The algorithm starts by having each node $u$ encode each of its $n$ messages into a $B$-bit codeword $\{C(m_{u,v})\}_{v \in V}$. Each codeword $C(m_{u,v})$ is then partitioned between an (ordered) subset $A(u,v) \subseteq V$ of size $B$, defined as follows.

\smallskip
\noindent \textbf{Step (1): Sending encoded messages to received sets.} Node $v_1$ (of smallest $\ID$) (locally) picks $B$ random numbers $r_1,\dots,r_B$ chosen uniformly and independently from $\{1,\dots,n\}$ with repetitions. These numbers are sent by $v_1$ to all nodes by applying Procedure $\ExtMatchingTransmission$ (see \Cref{cor:broadcast}). Each node locally defines the permutations $p_1,\dots,p_B:V \rightarrow V$ such that $p_i(v_j) = v_{(j+r_i) \pmod n}$ for any $i \in [B]$. These $B$ random numbers allow each vertex $u$ to determine the receiving-set $A(u,v)$ for each message $C(m_{u,v})$. Let, 
\begin{equation}\label{eq:message-set}
A(v)=A(u,v)=(p_1(v), \ldots, p_i(v), \ldots, p_B(v)).
\end{equation}
Note that the set $A(u,v)$ as mentioned in Eq. (\ref{eq:message-set}) is ordered (hence, it should not be ordered based on vertex IDs).
Also, note $A(u,v)$ is the same for every $u$, i.e., it depends only on $v$.  Each node $u$ then sends the codeword $C(m_{u,v})$ to $A(u,v)$ by applying Proc. $\Send(u,A(u,v),C(m_{u,v}))$. 
Let $\widetilde{C}(m_{u,v})$ be the received codeword at $A(u,v)$, where $\widetilde{C}(m_{u,v})[i]$ is the bit received by the node $p_i(v)$.

\smallskip
\noindent \textbf{Step (2): Sending encoded messages from received sets to targets.} The algorithm defines $B$ routing instances of the 
$\SuperMessagesRouting$ problem with parameters $k=1$ and $\lambda=n$, i.e., where each node holds a single $n$-bit message. Specifically, the $i^{th}$ instance $P_i$ is defined on the following input. Each node $w$ holds the single $n$-bit message, denoted as $m_{i}(w)$, whose target is $v=p^{-1}_i(w)$,\footnote{Here, $p^{-1}_i$ denotes the inverse permutation of $p_i$, i.e. $p^{-1}_i(p_i(w)) = w$ for all $w \in V$.} i.e.,
\begin{equation}\label{eq:mpi}
m_{i}(w)=\widetilde{C}_i(w)=(\widetilde{C}(v_1,v)[i],\ldots, \widetilde{C}(v_n,v)[i])~ \mbox{~and~} \Target(m_i(w)) =\{v\}.
\end{equation}
That is, the node $w$ holds the $i^{th}$ (possibly corrupted) bit of each of the $n$ encoded messages $\{C(m_{u,v})\}_{u \in V}$.

Each instance $P_i$ is solved by applying Alg. $\ExtMatchingTransmission$ with bandwidth parameter of $1$ (see \Cref{lem:static_learn_subprocedure} for the details). At the end of these $B$ applications, each $v$ obtains all the messages targeted at $v$, i.e., $\{\widetilde{C}_i(p^{-1}_i(v))\}_{i=1}^B$. This allows $v$ to locally compute its $n$ received codewords:

\begin{equation}
\label{eq:codeuv}
\widetilde{C}(u,v)=\widetilde{C}_1(u,v) \circ \ldots \circ \widetilde{C}_B(u,v)~, \forall u\in V~.
\end{equation}
Finally, $v$ sets $m'_{u,v}=\Decode(\widetilde{C}(u,v))$ for every $u$. In the analysis, we show that $m'_{u,v}=m_{u,v}$ w.h.p. This completes the description of the algorithm.

\paragraph{Analysis (Proof of \Cref{thm:non_rusing_main}).} We start by showing that the algorithm can be implemented in $O(1)$ rounds.

\begin{lemma}
\label{lem:static_first_round}
Applying Proc. $\Send(u,A(u,v),C(m_{u,v}))$ for each $u,v \in V$ can be implemented in parallel using one round.
\end{lemma}
\begin{proof}
    For each $i \in [B]$, each node sends the $i^{th}$ bit of $C(m_{u,v})$ to $p_i(v)$, the $i^{th}$ node of $A(u,v)$. Since $p_1,\dots,p_B$ are permutations, for any $i \in [B], w \in V$, node $w = p_i(v)$ for exactly one node $v \in V$. Therefore, each node $v$ sends exactly one bit to each node $w$ for each $i \in [B]$, and hence also receives exactly one bit from each node $w$ for each $i \in [B]$. The claim follows.
\end{proof}

\begin{lemma}
\label{lem:static_learn_subprocedure}
For each $i \in [B]$, instance $P_i$ can be solved in $O(1)$ rounds and using $1$ bit of bandwidth per edge (i.e., sending $1$ bit on each edge per round).
Hence, all $B$ instances can be implemented in $O(1)$ rounds, in parallel.
\end{lemma}
\begin{proof} 
For instance $P_i$, we define the following $\ExtMatchingTransmission$ instance: each node $v$ holds a single $n$-bit super-message $\Minput(v) = \{m_i(v)\}$ and has a single output message $\Moutput(v) = \{m_i(p_i(v))\}$. By \Cref{thm:generalized_routing_main}, Alg. $\ExtMatchingTransmission$ runs in $O(1)$ rounds and uses $1$-bit of bandwidth per edge.
\end{proof}

Next, we prove correctness of the algorithm:

\begin{lemma}\label{lem:distance-NR}[Correctness]
For any $u,v \in V$, w.h.p. it holds that (a) $\Hamm(\widetilde{C}(u,v),C(m_{u,v})) \leq 2\alpha n$ and (b) $m'_{u,v}=m_{u,v}$.
\end{lemma}
\begin{proof}
For any $u,v \in V, i \in [B]$, let $I_i(u,v)$ be the indicator random variable on the event that $C(m_{u,v})[i] \neq \widetilde{C}(m_{u,v})[i]$. First we show that $\Pr\left(I_i(u,v) = 1 \right) \leq \alpha$, and that for any $u,v \in V$, $I_1(u,v),\dots,I_B(u,v)$ are mutually independent random variables.

Recall that for any $i \in [B]$, $r_i$ is a uniformly random number in $\{1,\dots,n\}$, and therefore $p_i(v)$ is a uniformly random node in $\{v_1,\dots,v_n\}$. Additionally, recall that in the $\alpha$-NBD setting, in every round $j$, the adversary picks set the corrupted edge set $F_j$ only based on the initial configuration, and it is oblivious to the randomness of the algorithm. In particular, for the round $j$ in which the network runs $\Send(u,A(u,v),C(m_{u,v}))$ for all $u,v \in V$ in parallel, it holds that $\Pr\left((u,p_i(v)) \in F_j\right) \leq \deg_{F_j}(u)/n \leq \alpha$. Moreover, since $r_1,\dots,r_B$ are chosen independently, then $p_1(v),\dots,p_B(v)$ are independent random variables, since $p_i(v)$ only depends on $v$ and $r_i$.

Let $X_{u,v} = \sum_{i=1}^B I_i(u,v)$. By the Chernoff bound, $X_{u,v} \leq 2\alpha B$ w.h.p. (since $B = \Omega(\log{n})$). Therefore, w.h.p. $\Hamm(\widetilde{C}(u,v),C(m_{u,v})) \leq 2\alpha B$. Finally, we claim that $m'_{u,v}=m_{u,v}$, w.h.p. Since $\alpha \leq 1/(8 \cdot 10^4)$, we have that 
$$\Hamm(\widetilde{C}(m_{u,v}),C(m_{u,v})) \leq 2\alpha B < \delta_C B/2~.$$ 
Hence, $m'_{u,v} = \Decode(\widetilde{C}(m_{u,v})) = m_{u,v}$.
\end{proof}

\Cref{thm:non_rusing_main} follows immediately by \Cref{lem:static_first_round}, \Cref{lem:static_learn_subprocedure},  and \Cref{lem:distance-NR}. This completes the proof of \Cref{thm:non_rusing_main} (i.e., \Cref{thm:informal_non_adaptive}).

    \subsection{Handling Adaptive Adversaries}\label{sec:adaptive}

We present an $O(1)$-round randomized compiler in the $\alpha$-ABD setting, for $\alpha = \exp(-\sqrt{\log{n}\log\log{n}})$. Our algorithm combines both locally decodable codes (LDCs) and sketching techniques. The limitation on the fault parameter $\alpha$ stems from the parameters of the best known LDC, thus we present our theorem in terms of the LDC parameters, to allow the integration of future LDCs with improved parameters.
Missing proofs are deferred to \Cref{app:adaptive}. 
\begin{theorem}
\label{lem:mobile_main}
Let $C$ be a non-adaptive LDC with query complexity $q$, constant rate and relative distance, such that $\LDCDecode$ fails with probability at most $1/n^3$, and whose local decoding computation uses at most $O(n)$ random bits. There is a randomized protocol $\AdaptativeAllToAll$ that solves an $\AllToAllComm$ instance in the $\alpha$-ABD setting with $B=1$ for $\alpha=\Theta(1/(q\log^6{n}))$ in $O(1)$ rounds.
\end{theorem}
 \Cref{thm:informal_adaptive_randomized} is obtained from \Cref{lem:mobile_main} by using the LDC of \Cref{lem:LDC}. Throughout the algorithm, we assume w.l.o.g. that $1/\alpha$ is an integer. Otherwise, we can take $\alpha' < \alpha$ to be such that $\frac{1}{\alpha'} = \lceil \frac{1}{\alpha} \rceil$ and apply the algorithm on $\alpha'$ (which only affects constants). Additionally, by \Cref{lem:change_n_value}, we can assume w.l.o.g. that $n$ is divisible by $\frac{1}{\alpha}$.

\smallskip
\noindent \textbf{Description of Algorithm $\AdaptativeAllToAll$.}

\noindent \textbf{Step I. Initialization Phase.} In the first round, each node $u$ sends to each node $v$ the message $m_{u,v}$ (some of these messages could be corrupted). Denote by $\widetilde{m}_{u,v}$ the message that $v$ receives from $u$ in this step. We identify each message $m_{u,v}$ with $\ID(u,v)=\ID(u)\circ \ID(v)$. Note that despite the fact that the received message $\widetilde{m}_{u,v}$ might be corrupted, $v$ still knows the ID of this message $\ID(u,v)$ as it knows the sender name $u$ (as we assume the KT1 model). 

Following this, node $v_1$ generates two random strings $R_1,R_2$ of size $O(n)$ bits each, and broadcasts $R_1,R_2$ to the entire network using Proc. $\ExtMatchingTransmission$ (see \Cref{cor:broadcast}). For every node $v$, denote the subset of corrupted received messages by: 
$$\Mbad(v) = \{\widetilde{m}_{u,v} \stt u \in V, \widetilde{m}_{u,v} \neq m_{u,v}\}.$$ 

For the remainder of the algorithm, our goal is to allow each node $v \in V$ identifying and correcting the corrupted messages $\Mbad(v)$. Crucially, we remark that the random bits $R_1,R_2$ are determined only \emph{after} $\Mbad(v)$ has been chosen by the adversary.

For a subset of vertices $A \subseteq V$, let
$\Mbad(A,v) = \{m_{u,v} \neq \widetilde{m}_{u,v} \stt u \in A\}$ be the set of corrupted messages received by $v$ in this step from a node in $A$.

\smallskip
\noindent \textbf{Step II. Information Concentration Phase.} 
The goal of this step is to concentrate the information needed to correct the corrupted messages $\bigcup_{v \in V} \Mbad(v)$ on some ``small'' set of $\widetilde{\Theta}(\alpha n)$ nodes.   Information-theoretically, this is possible, since for any node $v$ it holds that $|\Mbad(v)| \leq \alpha n$, i.e., there are at most $\alpha n^2$ corrupted bits and each node can receive $O(n)$ bits in $O(1)$ rounds. To achieve this goal, we use sparse recovery sketches as a sort of a compression mechanism, that allows us to encode the information needed to correct $\Mbad(v)$ for each $v$ using less space, while being oblivious to the set $\Mbad(v)$.

\smallskip
\noindent \textbf{II(a). Node Partitionings.}
The algorithm defines two partitions of the nodes, given by $S$ and $P$. The first partition $S=\{S_1,\ldots, S_{1/\alpha}\}$ is obtained by partitioning $V$ deterministically into $1/\alpha$ consecutive subsets of size $\alpha n$, denoted $S_1 = \{v_1,\dots,v_{\alpha n}\}, S_2 = \{v_{\alpha n+1},\dots,v_{2\alpha n}\}, \dots S_{1/\alpha} = \{v_{n-\alpha n+1},\dots,v_n\}$. The second partition $P=\{P_1,\ldots, P_{\alpha n}\}$ with $|P_j|=1/\alpha$ is randomized and is computed locally by each node using the random coins of $R_1$, as described in the next lemma.

\begin{lemma}
\label{lem:bad_messages_in_random_partition}
Given the random string $R_1$, there is a zero-round protocol that outputs a random partition $P = \{P_1,\dots,P_{\alpha n}\}$, such that (a) $|P_i| = 1/\alpha$ for each $j \in [\alpha n]$, and (b) w.h.p. for all $j \in [\alpha n]$ and $v \in V$ it simultaneously holds that $|\Mbad(P_j,v)| = O(\log{n})$. The local computational time of the nodes is polynomial.
\end{lemma}
\def\APPENDRandomPartition{
\begin{proof}[Proof of Lemma \ref{lem:bad_messages_in_random_partition}]
    Each node does the following sequential sampling process, using the shared randomness $R_1$. We remark that while the size of $R_1$ is $\Theta(n)$ bits, in fact $\widetilde{\Theta}(1)$ random bits would suffice.
    
    Before describing the construction of $P$, we describe the construction of a partition $P' = \{P'_1,\dots,P'_{\alpha n}\}$ such that for each $j \in [\alpha n]$, both $|P'_j| = \Theta(1/\alpha)$ and $\Mbad(v,P'_j) = O(\log{n})$. Let $f:V \rightarrow [\alpha n]$ be a uniformly random function taken from a $\Theta(\log{n})$-wise independent hash function family (By \Cref{lem:kwise_indepedent_hash}, such a function can be computed using $\widetilde{O}(1)$ random bits). Let $P'_j = \{u \in V \stt f(u) = j\}$ be the set of nodes mapped into value $j \in [\alpha n]$. Since $|\Mbad(v)| \leq \alpha n$, then by the standard concentration bound on $\Theta(\log{n})$-wise independent variables (see \Cref{lem:concentration_k_wise}), it holds w.h.p. for every $j \in [\alpha n]$ that 
    \begin{enumerate}[(a)]
        \item $\frac{1}{2\alpha} \leq |P'_j| \leq \frac{2}{\alpha}$,
        \item  $|\Mbad(v,P'_j)| = O(\log{n})$.
    \end{enumerate}
    
    This concludes the construction of $P'$. Next, we describe how to adapt $P'$ to obtain a partition $P$ with set size exactly $1/\alpha$ while retaining the rest of the desired properties. We order the elements of $V$ by their value in the hash function $f(v)$, breaking ties arbitrarily. I.e., we assign a value to each node $\pi: V \rightarrow [n]$, such that for any $u_1,u_2 \in V$ with $f(u_1) < f(u_2)$, then $\pi(u_1) < \pi(u_2)$. Let $$P_j = \{u \stt  \frac{j-1}{\alpha}+1\leq \pi(u) \leq \frac{j}{\alpha}\}$$ be the $j^{th}$ segment of size $1/\alpha$ in the ordering $\pi$ of $V$. We remark that indeed, $|P_j| = 1/\alpha$. Moreover, since $P_j$ is comprised of consecutive elements of the ordering, and by the size guarantees on the sets of $P'$ (property (a)), for each $j \in [\alpha n]$ there exists a value $j' \in [\alpha n -2]$ such that $P_j \subseteq P'_{j'} \cup P'_{j'+1} \cup P'_{j'+2}$. By property (b), $$|\Mbad(v,P'_{j'})|+|\Mbad(v,P'_{j'+1})|+|\Mbad(v,P'_{j'+2})| = O(\log{n}),$$ meaning that $|\Mbad(v,P_{j})| = O(\log{n})$. 
\end{proof}
}

\smallskip
\noindent \textbf{II(b). Sketching and Routing.}
In this step, we perform the following three sub-steps: 
\begin{itemize}
\item For each $j \in [\alpha n]$ in parallel, each node $P_j[i]$ learns the set of messages $M(P_j,S_i)$ by applying Alg. $\ExtMatchingTransmission$ (see \Cref{lem:rand_mobile_transmission_first}).

\item Node $P_j[i]$ initializes for each $v \in S_i$ an empty $O(\log{n})$-sparse recovery sketch, denoted $\Sk(P_j,\{v\})$ using randomness $R_2$ (see \Cref{lem:k_sparse_recover}).  For each message in $m_{u,v} \in M(P_j,S_i)$, node $P_j[i]$ performs operation $\Add(\Sk(P_j,\{v\}),\ID(u,v) \circ m_{u,v},1,R_2)$, i.e., it adds the string $\ID(u,v) \circ m_{u,v}$ with frequency $1$ to the sketch $\Sk(P_j,\{v\})$. Crucially, we assume that all sketches are encoded by some fixed bit-length $t = O(\log^3 n)$. This can be assumed, w.l.o.g., via padding with zeros.

In the resulting state, node $P_j[i]$ has $|S_i| = \alpha n$ sketches, each encoded by exactly $t = \widetilde{O}(1)$ bits. For the remainder of the algorithm, for a subset $X=\{x_1,\ldots, x_\ell\} \subseteq V$ with $\ID(x_1)<\ldots< \ID(x_\ell)$, the concatenation of sketches of the messages $M(P_j,X)$ is denoted by 
\begin{equation}
\label{eq:sketch-string}
\Sk(P_j,X)=\Sk(P_j,\{x_1\})\circ \ldots \circ \Sk(P_j,\{x_\ell\})~.
\end{equation}
When $X=V$, we may omit it and simply write $\Sk(P_j)=\Sk(P_j,V)$. At this point $\Sk(P_j)$ is ``collectively held'' by the nodes of $P_j$. In particular, node $P_j[i]$ holds $\Sk(P_j,S_i)$. We make the following very simple observation, which follows from the fact that all sketches are encoded by the same length $t$. For every $i \in \{1,\ldots, n\}$ and $j \in \{1,\ldots, \alpha n\}$, it holds that:

\begin{equation}
\label{eq:sketch_location_in_string}
   \Sk(P_j,\{v_i\}) = \Sk(P_j)[(i-1)t+1] \circ \dots \circ \Sk(P_j)[i\cdot t]~. 
\end{equation}

The positions of the string $\Sk(P_j,\{v_i\})$ in $\Sk(P_j)$ is independent of $j$. 

Let $C$ be an LDC code $C$ with the properties stated in \Cref{lem:mobile_main} and let $\tau_C \in (0,1)$ denote its rate, define:
\begin{equation}\label{eq:x-def}
 x=\lfloor \tau_C \cdot n \rfloor - (\lfloor \tau_C \cdot n \rfloor \pmod t)~.   
\end{equation}

Note that $x$ is a multiple of $t$ and that $x/\tau_C\leq n$. For the purpose of distributing the codewords of the sketches $\Sk(P_j)$, for every $j \in \{1,\ldots, \alpha n\}$, the algorithm partitions the string $\Sk(P_j)$ into $b$ sub-strings 
$$\Sk(P_j)=\Sk_1(P_j)\circ  \ldots \circ \Sk_b(P_j) \mbox{~where~} b = \lceil n \cdot t/x \rceil~.$$ 
By padding $\Sk(P_j)$ with zeros, we guarantee that 
each $\Sk_\ell(P_j)$ has exactly $x$ bits for every $\ell \in \{1,\ldots, b\}$. Since we take  $x$ to be a multiple of $t$, we ensure that each sketch $\Sk(P_j, \{v_i\})$, for each $v_i \in V$, is fully contained in a single string $\Sk_\ell(P_j)$ (i.e. no sketch is ``cut in half'' between two $\Sk_\ell(P_j)$ strings). 

\item The algorithm applies Proc. $\ExtMatchingTransmission$ in order to send the $\ell^{th}$ sketch $\Sk_\ell(P_j)$ to the   $\ell^{th}$ node in $P_j$ namely, $P_j[\ell]$, for every $\ell \in \{1,\ldots, b\}$ (see \Cref{lem:rand_mobile_transmission_second}). We denote the $b$ nodes $P_j[1],\ldots, P_j[b]$ as the \emph{leaders} of the group $P_j$. We then have that each sketch $\Sk(P_j)$ is held jointly by the $b=\widetilde{O}(1)$ leaders of $P_j$ for every $j \in [\alpha n]$.

\end{itemize}
Overall, at the end of this step, the concatenation of all sketches $\Sk(V)$, is jointly held by $b \cdot \alpha \cdot n$ nodes, as desired.

\smallskip
\noindent \textbf{Step III. Learning the Messages Sketches via LDC.}
Our goal in this step is for each $v$ to learn the sketch of its received messages $\Sk(V,\{v\})=\Sk(P_1,\{v\})\circ \ldots \circ \Sk(P_{\alpha n},\{v\})$. Then, by using the sparse-recovery properties of these sketches, $v$ will be able to deduce the correction of its corrupted messages $\Mbad(v)$. To learn the sketches, the algorithm uses LDCs.

Each sketch $\Sk_\ell(P_j)$ will be encoded using the LDC code $C$ with the properties stated in \Cref{lem:mobile_main}. The bits of each codeword will then be distributed among all the nodes. Each $v$ will compute 
$\Sk(P_j,\{v\})$ by querying a few nodes using the local decoding properties of the LDC. We apply the following steps:

 \begin{itemize}
    \item For each $\ell \in [b]$, node $P_j[\ell]$ encodes $\Sk_\ell(P_j)$ using the LDC $C$ into a codeword $C(\Sk_\ell(P_j))$ of size $x/\tau_C\leq n$. For simplicity, we assume again that each $C(\Sk_\ell(P_j))$ has exactly $n$ bits, this can be done by padding with zeros.

    In a single round of communication, each node $P_j[\ell]$ distributes its codeword $C(\Sk_\ell(P_j))$ among all the nodes in $V$, by sending the $r^{th}$ bit of $C(\Sk_\ell(P_j))$ to node $v_r$. This is done in parallel for every $j \in [\alpha n]$ and $\ell \in [b]$. Let
$\widetilde{C}(j,\ell,r)$ be the bit received by $v_r$ from the leader node $P_j[\ell]$. Also, $\widetilde{M}(v_r)$ be the concatenation of all the received bits $\{\widetilde{C}(j,\ell, r)\}_{j \in [\alpha n], \ell \in [b]}$ ordered by their index $(j,\ell)$. 
Formally, define
$$\widetilde{M}(v_r)=\widetilde{C}(1,1,r)\circ \ldots \circ\widetilde{C}(1,b,r) \circ \ldots \circ \widetilde{C}(\alpha n,1,r) \circ \ldots \circ \widetilde{C}(\alpha n,b,r)~.$$
Denote $M(v_r)[j,\ell]=\widetilde{C}(j,\ell,r)$. Clearly, $\widetilde{M}(v_r)$ has $b \cdot \alpha n$ bits. Our goal is to send each message $\widetilde{M}(v_r)$ to a subset of nodes $\Target(\widetilde{M}(v_r))$, such that each node $v_i \in V$ will be able to determine its sketch $\Sk(V,\{v_i\})$ by extracting information from its received $\widetilde{M}(v_r)$ messages. To keep the number of rounds small, it is desired that each $v_i$ will be a target of a small number of $\widetilde{M}(v_r)$ messages. For this purpose we exploit the local decoding properties of the LDC code. 

\item Node $v_1$ generates a random string $R_3$ of size $O(n)$ bits, and broadcasts $R_3$ to the entire network using Proc. $\ExtMatchingTransmission$  (see \Cref{cor:broadcast}).

\item This step is aimed at defining the targets set $\Target(\widetilde{M}(v_r))$ of the messages $\{\widetilde{M}(v_r)\}_{v_r \in V}$. This calls for characterizing the precise locations of the bits in the codewords $\{C(\Sk_\ell(P_j))\}_{\ell \in [b], j \in [\alpha n]}$ that a node $v_i$ should learn in order to deduce the original sketches $\{\Sk(P_j,\{v_i\})\}_{j \in [\alpha n]}$.

For every node $v_i \in V$ and $j \in [\alpha n]$ let $\ell_i \in [b]$ be the index of the leader node in $P_j$ that holds $\Sk(P_j,\{v_i\})$. Recall that each individual sketch $\Sk(P_j,\{v_i\})$ is held entirely by some leader node $P_j[\ell]$ for every $j$. By Eq. (\ref{eq:sketch_location_in_string}), it holds that the identity of $\ell$ is the same across all $j \in [\alpha n]$, hence $\ell_i$ is well defined. By knowing $n, x$ and $t$, each node $v_i$ can deduce the $t$ consecutive indices in $\Sk_{\ell_i}(P_j)$ that corresponds to the sketch $\Sk(P_j,\{v_i\})$. Note that these indices do not depend on $j$, we therefore denote them by $\loc(v_i)=\{z_{i,1},\ldots, z_{i,t}\}$ where 
$$\Sk(P_j, \{v_i\})=\Sk_{\ell_i}(P_j)[z_{i,1}]\circ \ldots \circ \Sk_{\ell_i}(P_j)[z_{i,t}], \mbox{~~for every~} j \in [\alpha n]~.$$

The goal of each $v_i$ is to query each of the codewords $C(\Sk_{\ell_i}(P_j))$, for $j \in [\alpha n]$, at $q t$ locations in order to determine $\Sk_{\ell_i}(P_j)[k]$ for every $k \in \loc(v_i)$, hence recovering $\Sk(P_j,\{v_i\})$, for every $j \in [\alpha n]$.

Altogether, $v_i$ should deduce $\alpha \cdot n \cdot q \cdot t$ bits that are stored by the set of $q \cdot t$ nodes given by:
$$N(v_i)=\{ v_r ~\mid~ k \in \loc(v_i),  r\in \LDCDecodeIndices(k,R_3)\}.$$

\item Using Proc. $\ExtMatchingTransmission$, the algorithm sends the $b\alpha n$-bit message $\widetilde{M}(v_r)$ to all nodes in $\Target(\widetilde{M}(v_r))=\{v_i ~\mid~ v_r \in N(v_i)\}$, for every $v_r \in V$. Overall, each node $v_i$ is the target of $q \cdot t \cdot b \cdot \alpha n$ bits, and each $v_r \in V$ is required to send $b \alpha n$ bits to possibly $n$ nodes. By taking $\alpha\leq 1/(q \cdot t\cdot b)$, we show in the analysis, that this can be done in $O(1)$ rounds (see \Cref{lem:rand_mobile_transmission_third}). 

\item At this point each $v_i$ holds the messages $\{\widetilde{M}(v_r)\}_{r \in N(v_i)}$. Each sketch $\Sk(P_j,\{v_i\})$ can be obtained by reading the bits $\{\widetilde{M}(v_r)[j,\ell_i]\}_{r \in N(v_i)}$ and applying the decoding algorithm $\LDCDecode$ on the obtained bits. 
\end{itemize}

\smallskip
\noindent \textbf{Step IV. Correction Step.}
In this final step, each node $v$ initially holds $\Sk(V,\{v\})$, and its goal is to identify and correct $\Mbad(v)$. This step is done completely locally by the nodes, i.e., without any communication.

For each $j \in [\alpha n]$, each $v$ adds for each $u \in P_j$ the message $\widetilde{m}_{u,v}$ with frequency $-1$ to $\Sk(P_j,\{v\})$, i.e., it performs $\Add(\Sk(P_j,\{v\}),\ID(u,v) \circ \widetilde{m}_{u,v},-1,R_2)$. It then applies the $\Recover(\Sk(P_j,\{v\}))$ operation each of $\Sk(P_j,\{v\})$ to receive a set of messages $M(P_j,\{v\})$. For any $u,v \in V$, if a message $\ID(u,v) \circ m'_{u,v} \in M(v)$ is recovered by $v$ with frequency $1$, $v$ outputs $m'_{u,v}$ as the message sent by $u$ to $v$. Otherwise, i.e. if no such message was recovered, node $v$ outputs $\widetilde{m}_{u,v}$ instead. This concludes the description of the algorithm.

\begin{figure}
    \centering
    \includegraphics[width=0.8\textwidth]{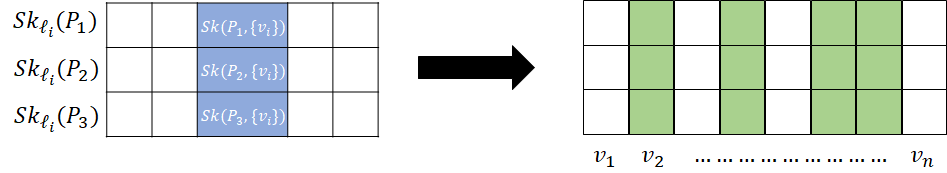}
    \caption{Node $v_i$ wishes to decode all bits of $\Sk(P_j,\{v_i\})$ from $\Sk_{\ell_i}(P_j)$ (marked as the blue cells). Crucially, these indices are in the same positions $P_j$. Since the LDC is non-adaptive, for a node $v_i$ to be able to decode all blue positions  simultaneously, its sufficient to learn the codeword bits given to nodes $\{v_r\}_{r \in N(v_i)}$ (marked as the green cells - due to non-adaptivity of $\LDCDecode$, they are the same positions for all $P_j$ if we use the same randomness to decode them all).} 
    \label{fig:random_ldc_intuition}
\end{figure}

\smallskip
\noindent \textbf{Analysis.} First, we describe the details of the routing lemmas used in the algorithm. For the purposes of this section, we say that a node $v$ learns $M(A,B)$ for $A,B \subseteq V$ if it learns each message $m_{a,b} \in M(A,B)$ and knows associated ids, $\ID(a,b)$. We remark that when using $\ExtMatchingTransmission$, each node knows the super-message identifier $(u,j)$ of each super-message it receives. Hence, for example, in a $\ExtMatchingTransmission$ instance where the $j^{th}$ input super-message of a node $u$ is of the form $M^{\circ}(A,B)$, each node $v$ learning the super-message can deduce both the content and the associated identifiers of each message in $M(A,B)$. We omit the exact specifics of these details from the proofs. Missing proofs are deferred to \Cref{app:adaptive}.
\begin{lemma}
    \label{lem:rand_mobile_transmission_first}
    At the first part of \emph{Information Concentration Step},  for each $j \in [\alpha n], i \in [1/\alpha]$, each node $P_j[i]$ can learn $M(P_j,S_i)$ in $O(1)$ rounds.
\end{lemma}
\begin{proof}
    We use Proc. $\ExtMatchingTransmission$ on a routing instance where each node holds $1/\alpha$ super-messages, each of size $\alpha n$. Specifically, each node $v$ defines $1/\alpha$ super-messages, 
    $$\Minput(v) = \{M^{\circ}(\{v\},S_1),\ldots,M^{\circ}(\{v\},S_{1/\alpha})\}.$$ 
    
     For $j \in [\alpha n]$ and $i \in [1/\alpha]$, the super-messages targeted at node $P_j[i]$ are: 
     $$\Moutput(P_j[i]) = \{M^{\circ}(\{v\},S_i)\}_{v \in P_j}.$$ 
     
     Recall that $|S_i| = \alpha n$ and $|P_j|=1/\alpha$. Hence, each super-message has $O(\alpha n)$ bits, each node has $1/\alpha$ input super-messages in $\Minput(v)$, and it is a target of $1/\alpha$ output super-messages in $\Moutput(v)$. Therefore, \Cref{thm:generalized_routing_main} guarantees that Proc. $\ExtMatchingTransmission$ terminates correctly after $O(1)$ rounds.
\end{proof}

\begin{lemma}
\label{lem:rand_mobile_transmission_second}
    At the second part of Step (II), for each $j \in [\alpha n], \ell \in [b]$, each node $P_j[\ell]$ can learn $\Sk_\ell(P_j)$ in $O(1)$ rounds.
\end{lemma}
\begin{proof}
    We show that we apply Proc. $\ExtMatchingTransmission$ on a routing instance where each node has as a single input super-message of $O(\alpha nt)$ bits, and each node is a target of $O(1/(\alpha t))$ super messages. For any $i \in [1/\alpha], j \in [\alpha n]$,  $P_j[i]$ defines a single super-message with $O(|S_i| \cdot t)=O(\alpha t n)$ bits:
    $$\Minput(P_j[i]) = \{\Sk(P_j,S_i)\}~.$$

    We now define the desired output set of super-messages for every leader node $P_j[\ell]$ for $\ell \in [b]$ and $j \in [\alpha n]$. Let $V_\ell=\{v_i ~\mid~ \ell=\ell_i\}$ be the set of nodes whose sketch $\Sk(P_j,\{v\})$ is held as part of $Sk_\ell(P_j)$ and recall that $V_\ell$ does not depend on $P_j$. It is then required for $P_j[\ell]$ to obtain the super-messages $\Sk(P_j,S_i)$ for every $S_i$ that contains some node in $V_\ell$. Formally, 
    $$\Moutput(P_j[\ell]) = \{\Sk(P_j,S_i) ~\mid~ i \in [1/\alpha], S_i \cap V_\ell \neq \emptyset\}~.$$
    Note that since each sketch $\Sk(P_j,\{v\})$ has exactly $t$ bits, given $\Moutput(P_j[\ell])$, $P_j[\ell]$ will be able to recover $Sk_\ell(P_j)$. For any non-leader node, $P_j[\ell']$ for $\ell' \geq b+1$, the set $\Moutput(P_j[\ell]')$ is empty. We next show that this routing can be done in $O(1)$ rounds. 

    Overall, each node holds a single input super-message with $O(\alpha t n)$ bits. Moreover, $P_j[\ell]$ is the target of at most 
    $$O(|V_\ell|/(\alpha \cdot n))=O(x /(t\alpha n)=O(1/(t\alpha)),$$
    super-messages, where the first equality follows by the fact that $Sk_\ell(P_j)$ has $x$ bits, and each node sketch consumes exactly $t$ bits; the last equality follows by the definition of $x$ in Eq. (\ref{eq:x-def}). Hence, by \Cref{thm:generalized_routing_main}, Proc. $\ExtMatchingTransmission$ is implemented in $O(1)$ number of rounds. (Note that we even have the budget to let each node also hold $O(1/(\alpha \cdot t))$ super-messages each of size 
    $O(\alpha t n)$ bits.)
\end{proof}

\begin{lemma}
\label{lem:rand_mobile_transmission_third}
    At Step (III), each node $v$, can learn the messages $\{\widetilde{M}(v_r) \mid r \in N(v)\}$ in $O(1)$ rounds.
\end{lemma}
\begin{proof}

    We use Proc. $\ExtMatchingTransmission$ on a routing instance where each node has as input one super-message of size $O(n/qt)$ bits, and each node outputs $O(qt)$ super messages. 

    For each node $v$ we define a single input super-message $\Minput(v) = \{\widetilde{M}(v)\},$ and define the output set 
    $\Moutput(v) = \{\widetilde{M}(v_r) \mid r \in N(v)\}.$ Since $\widetilde{M}(v)$ contains one bit for each $i \in [b]$ and $j \in [\alpha n]$, the total size of the super-message is $O(\alpha nb) = O(n/(qt))$, where the inequality follows by the fact that $b = O(nt/n) = O(t)$, and by choice of $\alpha$ such that $\alpha = O(1/(qt^2))$. 
    
    Next, we bound $|\Moutput(v)|$. Since $\loc(v)$ is of size $|\Sk(V,\{v\}| = t$, it holds that $|N(v)| \leq qt$. Hence, $|\Minput(v)| = 1,|\Moutput(v)| \leq qt$.  Therefore, \Cref{thm:generalized_routing_main} guarantees that Proc. $\ExtMatchingTransmission$ terminates correctly after $O(1)$ rounds.
\end{proof}

Next,we show that each node $v$ can deduce $\Sk(V,\{v\})$ from $\{\widetilde{M}(v_r) \mid r \in N(v)\}$. We introduce the following notation. Recall that $\widetilde{C}(j,\ell,r)$ denotes the bit of $C(\Sk_\ell(P_j))$ that node $v_r$ receives from $P_j[\ell]$ (which may arrive corrupted) for every $r \in [n]$. Denote $\widetilde{C}(\Sk_\ell(P_j))$ as the received codeword distributively held by all nodes $V$. Formally:

    $$\widetilde{C}(\Sk_\ell(P_j))=\widetilde{C}(j,\ell,1)\circ \ldots \circ\widetilde{C}(j,\ell,n).$$

\begin{lemma}
\label{lem:randomized_codeword_distance}
    For all $j \in [\alpha n]$ and $\ell \in [b]$, $\Hamm(\widetilde{C}(\Sk_\ell(P_j)),C(\Sk_\ell(P_j))) \leq \alpha n.$
\end{lemma}
\begin{proof}
    Recall that the codeword $C(\Sk_\ell(P_j))$ is sent by node $P_j[\ell]$ to the rest of the network in a single round, where the $r^{th}$ bit of the codeword is sent to $v_r$. Therefore, the adversary can only corrupt at most $\alpha n$ of these sent bits. In other words,
    $$\Hamm(\widetilde{C}(\Sk_\ell(P_j)),C(\Sk_\ell(P_j))) \leq \alpha n.$$
\end{proof}

\begin{lemma}
\label{lem:randomized_learned_the_skeches}
    At the end of Step (III), each node $v_i$, can locally compute $\Sk(V,\{v_i\})$ from $\{\widetilde{M}(v_r) \mid r \in N(v_i)\}$ w.h.p.
\end{lemma}
\begin{proof}
    We show that node $v_i$ can compute each of the bits of $\Sk(V,\{v_i\})$. In other words, we show it can compute $\Sk_{\ell_i}(P_j)[k]$ for each $k \in \loc(v_i)$, and $j \in [\alpha n]$.
    
    To obtain $\Sk_{\ell_i}(P_j)[k]$, node $v_i$ locally applies $\LDCDecode(k,\widetilde{C}(\Sk_{\ell_i}(P_j),R_3)$, the local decoding algorithm with index $k$, and faulty codeword $\widetilde{C}(\Sk_{\ell_i}(P_j))$, using randomness $R_3$. By definition of $N(v_i)$, all the values of the queried indices chosen by  $\LDCDecodeIndices(k,R_3)$ are contained in $\{\widetilde{M}(v_r)[j,\ell_i] \mid r \in N(v_i)\}$. Therefore, $v_i$ holds all query values locally, and can perform the algorithm. It remains to show that the LDC decoding succeeds w.h.p.

    By \Cref{def:LDC}, $\LDCDecode$ succeeds w.h.p. under the assumption that the fraction of faults in the decoded codeword is at most $\delta_C/2$. This follows immediately by \Cref{lem:randomized_codeword_distance}, and the fact that the codeword length (before padding) is of size $\Theta(n)$, while $\alpha = o(1)$, and $\delta_C$ is some constant. 

    Therefore $v_i$ can decode the $k^{th}$ bit of $\Sk_{\ell_i}(P_j)$ correctly w.h.p., and by repeating this process for all $k \in \loc(v_i)$ and $j \in [\alpha n]$, it obtains $\Sk(V,\{v_i\})$ w.h.p.
\end{proof}
The remaining proof of correctness follows a very similar line along the algorithm of \cite{FP23}. For completeness, we provide it here (see \Cref{app:adaptive}).

\begin{lemma}\label{lem:correct-adaptive}
    Each node $v$ outputs $M(V,\{v\})$ w.h.p.
\end{lemma}
\def\APPENDADAPCORRECT{
\paragraph{Proof of Lemma \ref{lem:correct-adaptive}.}
We start by showing the following lemma.
\begin{lemma}
    \label{lem:sketch_content}
    Let $j \in [\alpha n]$, $u \in P_j$, and $v \in V$. In Step (IV), when Procedure $\Recover$ is applied, the frequency of $\ID(u,v) \circ m_{u,v}$ in $\Sk(P_j,\{v\})$ is $1$ if $m_{u,v} \in \Mbad(u)$, and $0$ otherwise. Similarly, the frequency of $\ID(u,v) \circ \widetilde{m}_{u,v}$ in $\Sk(P_j,\{v\})$ is $-1$ if $m_{u,v} \in \Mbad(u)$, and $0$ otherwise. All other strings have frequency $0$ in $\Sk(P_j,\{v\})$.
\end{lemma}
\begin{proof}
    Let $f(s)$ denote the frequency of a string $s$ in $\Sk(P_j,\{v\})$. For each $u \in P_j$, $\ID(u,v) \circ m_{v,u}$ is added into $\Sk(P_j,\{v\})$ with frequency $1$, and $\ID(u,v) \circ \widetilde{m}_{v,u}$ is added with frequency $-1$.  
    
    Since identifiers are unique, we get that only the insertion of $\ID(u,v) \circ m_{v,u}$ can affect the frequency of $\ID(u,v) \circ \widetilde{m}_{v,u}$, and vice versa.
    
    If $\widetilde{m}_{u,v} \notin \Mbad(P_j,u)$, then $\widetilde{m}_{u,v} = m_{u,v}$ and the total frequency of both is $f(\ID(u,v) \circ \widetilde{m}_{u,v}) = f(\ID(u,v) \circ m_{u,v}) = 0$. If $\widetilde{m}_{u,u} \in \Mbad(P_j,u)$, then $\widetilde{m}_{u,v} \neq m_{u,v}$. Therefore $f(\ID(u,v) \circ \widetilde{m}_{u,v}) = 1$ and $f(\ID(u,v) \circ m_{u,v}) = -1$. Since no other string is added into $\Sk(P_j,\{v\})$, the claim follows.
\end{proof}

We are now ready to complete the proof of \Cref{lem:correct-adaptive}. 

\begin{proof}[Proof of \Cref{lem:correct-adaptive}]
    By \Cref{lem:randomized_learned_the_skeches}, each node $v$ learns $\Sk(V,\{v\})$, i.e. it knows $\Sk(P_j,\{v\})$ for all $j \in [\alpha n]$. By \Cref{lem:bad_messages_in_random_partition}, $|\Mbad(P_j,v)| = O(\log{n})$ with high probability. Therefore, by \Cref{lem:sketch_content}, there are at most $O(\log{n})$ elements with non-zero frequency in $\Sk(P_j,\{v\})$. This guarantees that $\Recover(\Sk(P_j,\{v\}))$ returns all elements and their frequencies with high probability. By \Cref{lem:sketch_content}, for $u \in P_j$, the message $\ID(u,v) \circ m_{u,v}$ has frequency $1$ in $\Sk(P_j,\{v\})$ if and only if $\widetilde{m}_{u,v} \in \Mbad(v)$. 
    
    We show that $v$ outputs $m_{u,v}$ for all $u \in V$. For this, we split into two cases: if $\widetilde{m}_{u,v} \in \Mbad(v)$, then $\ID(u,v) \circ m_{u,v}$ has frequency $1$ in its respective sketch, hence $v$ outputs $m_{u,v}$. If $\widetilde{m}_{u,v} \notin \Mbad(v)$, then $\ID(u,v) \circ m_{u,v}$ has frequency $0$ and node $v$ outputs correctly $\widetilde{m}_{u,v} = m_{u,v}$.
\end{proof}
}

This completes the proof of \Cref{lem:mobile_main} (i.e.,  \Cref{thm:informal_adaptive_randomized}).

    \section{Deterministic Compilers}\label{sec:deterministic_compilers}

    In this section we present deterministic algorithms for the $\AllToAllComm$ problem against $\alpha$-ABD adversaries. \Cref{sec:det-constant} assumes constant $\alpha$ and \Cref{sec:det-nonconstant} assumes $\alpha=O(1/\sqrt{n})$.
    \subsection{Deterministic Algorithm for $\alpha=\Theta(1)$}\label{sec:det-constant}
In this section, we show an $O(\log{n})$ deterministic algorithm for the $\AllToAllComm$ problem in the presence of $\alpha$-ABD adversary for $\alpha=O(1)$.

\begin{theorem}
    \label{lem:det_mobile_high_noise_main}
    Let $\alpha \leq 1/(16 \cdot 10^4)$. There is a deterministic protocol solving $\AllToAllComm$ in the $\alpha$-ABD setting in $O(\log{n})$ rounds, for bandwidth $B = 1$.
\end{theorem}

By \Cref{lem:change_n_value}, we can assume that $n$ is a power of two, i.e. $n = 2^\ell$ for some integer $\ell$. This comes at the cost of reducing the fault parameter $\alpha$ by a constant factor. As the value of $\alpha$ is dictated by the super-message routing of \Cref{thm:generalized_routing_main} and \Cref{lem:change_n_value}, the correctness holds for $2\alpha \leq 1/(8 \cdot 10^4)$. Throughout the section, we assume w.l.o.g. that the nodes have unique identifiers in $\{0,\dots,2^\ell-1\}$. This can be assumed since in our model, each node initially knows all identifiers (i.e, the KT1 model), and hence can locally rename identifiers according to lexicographic order. Denote $\ID(v)[i]$ to be the $i^{th}$ bit of $\ID(v)$. For $v \in V,i \in [\log{n}],b \in \{0,1\}$, define $\Flip(v,i,b)$ to be the node $u$ with $\ID(u)[i] = b$ and $\ID(u)[j] = \ID(v)[j]$ for any $j \neq i$. 

As previous sections, we denote for message $m_{u,v}$ it identifiers as $\ID(m_{u,v}) = \ID(u,v) = \ID(u) \circ \ID(v)$. Additionally, we denote $\Source(m_{u,v}) = u$, and $\Target(m_{u,v}) = v$. For a set of messages $M$, denote by $\ID(M) = \{\ID(m) \stt m \in M\}$ as the set of identifiers of $M$. If the set $M$ is ordered, then $\ID(M) = (\ID(m) \stt m \in M)$.

We define a \emph{target-source identifier ordering} of a set of messages, where we sort by the target id in descending order, breaking ties according to the source id in descending order. This is defined formally as follows:

\begin{definition}
\label{def:target_source_ordering}
    For a set of messages $M \subseteq M(V,V)$, we define the target-source identifier ordering as an assignment of unique ranks $\rk(m_{u,v}) \in [|M|]$ for each message $m_{u,v} \in M$ such that for any $m_{u,v},m_{u',v'} \in M, \rk(m_{u,v}) > \rk(m_{u',v'})$ if either $\ID(v) > \ID(v')$, or if both the conditions $\ID(v) = \ID(v')$ and $\ID(u) > \ID(u')$ hold.
\end{definition}

\begin{figure}
    \centering
    \includegraphics[width=0.8\textwidth]{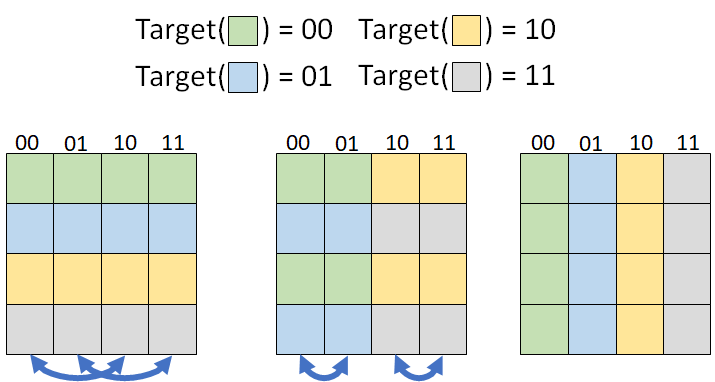}
    \caption{An illustration of the algorithm of \Cref{lem:det_mobile_high_noise_main} for $n=4$. Each matrix corresponds to the set of messages a node holds in a given step, where each column $i$ corresponds to the set of messages node $v_i$ holds. Blue arrows correspond to the message exchanges which occur in this step. Different colors of rectangles correspond to different targets, as written in the legend above. Initially each node has a message for each node. In the start of the second step, the first two nodes collectively hold the messages addressed to the first two nodes, and the last two nodes have the symmetric property. Finally, after two steps, all nodes hold their target messages.}
    \label{fig:det_high_noise}
\end{figure}

\noindent \textbf{Algorithm Description}
The protocol proceeds in $i=1,\dots,\log{n}$ iterations. Each iteration consists of a local computation step, and a communication step, as described next:

\smallskip
\noindent \textbf{Initialization Step.} At the start of $i = 1$, each node $v$ locally sets $M_1(v) = M(\{v\},V)$ to be $v$'s input messages. We assume $M_1(v)$ is ordered according to target-source identifier ordering (i.e. by descending order of target identifier, ties broken by source identifier; see \Cref{def:target_source_ordering}). This concludes the initialization step.

\smallskip
\noindent \textbf{Iteration $i \geq 1$  (Local Computation):} At the start of each iteration $i$, each node $v$ holds a set of messages $M_i(v)$. Throughout the algorithm, we maintain the invariant that nodes know the identifiers of all messages in $M_i(v)$, and that $M_i$ is sorted by target-source identifier ordering. 

Let $\rk(m)$ be the place of $m$ in the target-source identifier ordering of $M_i(v)$. 
The node $v$ splits $M_i(v)$ into two sets $M_i(v) = M^-_i(v) \cup M^+_i(v)$, where:
$$M^{-}_i(v) = \{m \in M_i(v) ~\stt~ \rk(m) \leq |M_i(v)|/2\} \mbox{~and~} M^{+}_i(v) = \{m \in M_i(v) ~\stt~ \rk(m) > |M_i(v)|/2\},$$ i.e. $M^-_i(v),M^+_i(v)$ are the messages in the lower and upper half of the ordering respectively.

Similar to prior sections, we denote $M^{\circ-}_i(v)$ and $M^{\circ+}_i(v)$ as the strings corresponding to these sets, ordered in the target-source identifier ordering. See Eq. (\ref{eq:message-concatentation}).

\smallskip
\noindent \textbf{Iteration $i \geq 1$ (Communication Step):} Using $\ExtMatchingTransmission$, each node $v$ sends the set $M^-_i(v)$ to node $\Flip(v,i,0)$ and the set $M^+_i(v)$ to node $\Flip(v,i,1)$ (See \Cref{lem:det_high_transmission} for details). From this procedure, each node $v$ receives two sets of messages $Q_{i,1}(v),Q_{i,2}(v)$, and sets the variable $M_{i+1}(v) = Q_{i,1}(v) \cup Q_{i,2}(v)$, sorted by target-source identifier ordering, to be the set of all received messages.\footnote{Note that $v$ is either $\Flip(v,i,0)$ or $\Flip(v,i,1)$, hence it ``receives'' one message set from itself, and one from another node.} We remark that node $v$ indeed knows the identifiers of the ordered sets $Q_{i,1},Q_{i,2}$ (see \Cref{lem:det_high_noise_char_complete}).

\smallskip
\noindent \textbf{Output:} At the end of iteration $i = \log{n}$, each node $v$ holds $M_{\log{n}+1}(v)$. In the analysis, we show that $M_{\log{n}+1}(v) = M(V,\{v\})$ (see \Cref{lem:det_high_noise_char_complete}). Each $v$ sets the output $m'_{u,v}$ to be the unique message $m \in M_{\log{n}+1}(v)$ with $\ID(m) = \ID(u,v)$. This concludes the description of the algorithm.

\smallskip
\noindent \textbf{Analysis.}
We denote the concatenation between two strings $s,s'$ as $s \circ s'$. Let $\Prefix(v,i)$ as the prefix of $\ID(v)$ up to the $i^{th}$ position, i.e. $\Prefix(v,i) = \ID(v)[1] \circ \dots \circ \ID(v)[i]$. Similarly, define $\Suffix(v,i)$ as the suffix of $\ID(v)$ up to the $i^{th}$ position from the end, i.e. $\Suffix(v,i) = \ID(v)[\log{n}-i] \circ \dots \circ \ID(v)[\log{n}]$. In particular, we define $\Prefix(v,0),\Suffix(v,0) = \emptyset$.

\noindent For an integer $i$ and a node $u$, define
$$P(u,i)=\{v \in V \mid \Prefix(v,i-1) = \Prefix(u,i-1)\}$$ and 
$$S(u,i)=\{v \in V \mid \Suffix(v,\log{n}-i+1) = \Suffix(u,\log{n}-i+1)\}.$$
Namely, these are the sets of nodes whose IDs agree with $\ID(u)$ on the first $i$ bits (resp., last $\log n-i+1$ bits). 
Note that $P(u,i)\cap S(u,i)=\{u\}$, $P(i,1)=V$ and $S(u,1)=\{u\}$. Also, observe that $S(u,\log n+1)=V$ and $P(u,\log n+1)=\{u\}$. The correctness of the algorithm follows by showing the following characterization of the message set $M_i(u)$ for every $u \in V$ and $i \in \{1,\ldots, \log n+1\}$.

\begin{lemma}
\label{lem:det_high_noise_char_complete}
For every $u \in V$ and $i \in [\log n+1]$, it holds that 
$M_i(u)=M(S(u,i),P(u,i))$. Hence, $M_{\log n+1}(u)=M(V, \{u\})$.
\end{lemma}
\begin{proof}
    We prove the claim by induction on $i$, for $1 \leq i \leq \log{n}+1$.
    
\noindent    \textbf{Base Case:} For $i = 1$, $P(i,1)=V$ and $S(u,1)=\{u\}$ and the claim holds. \smallskip
   
    \noindent \textbf{Induction Step:} Assume that the claim holds up to $i-1$, we next show that it holds for $i$. Fix $u \in V$ and assume that $\ID(u)[i-1] = 0$. The proof for the case where $\ID(u)[i-1] = 1$ is symmetric. 
    Let $u'=\Flip(u,i-1,1)$. By the induction assumption for $i-1$, it holds that:
    $$M_{i-1}(u)=M(S(u,{i-1}),P(u,{i-1})) \mbox{~~~and~~~} M_{i-1}(u)=M(S(u',{i-1}),P(u',{i-1}))~.$$

    By the definition of $u'$, we have:
    \begin{itemize}
        \item $\Prefix(u,i-2)=\Prefix(u',i-2)$, hence $P(u,i-1)=P(u',i-1)$,
        \item $\Prefix(u,i-1)=\Prefix(u,i-2)\circ 0$ and $\Prefix(u',i-1)=\Prefix(u,i-2) \circ 1$.
        \item $P(u,i-1)$ is the disjoint union of the two equally size sets $P(u,i)$ and $P(u',i)$.
    \end{itemize}

 Since all the IDs in $P(u,i)$ are strictly smaller than the IDs in $P(u',i)$ and as $|P(u,i)|=|P(u',i)|$, by sorting the messages in $M_{i-1}(u)$ and $M_{i-1}(u')$ based on target-source ordering, we get that:  
 $$M^{-}_{i-1}(u)=M(S(u,{i-1}),P(u,{i}))~~~\mbox{and}~~~~~M^{-}_{i-1}(u')=M(S(u',{i-1}),P(u,{i}))~.$$

Finally, observe that $S(u,i)$ is the disjoint union of $S(u,i-1)$ and $S(u',i-1)$. This holds since 
$\Suffix(u',i)=1 \circ \Suffix(u,i-1)$ and $\Suffix(u,i)=0 \circ \Suffix(u,i-1)$. Altogether, we have that $M_i(u)=M(S(u,i),P(u,{i}))$ as desired. 
\end{proof}

\begin{lemma}
\label{lem:det_high_transmission}
Using Proc. $\ExtMatchingTransmission$, each node may send $M^-_i(u)$ to node $\Flip(u,i,0)$ and $M^+_i(u)$ to node $\Flip(u,i,1)$ in $O(1)$ rounds.
\end{lemma}
\begin{proof}
By Lemma \ref{lem:det_high_noise_char_complete}, the total number of messages in $M_i(u)$ is $O(n)$. To see this note that $|S(u,i)|=2^{i-1}$ and $|P(u,i)|=2^{\log n-i+1}$, hence 
$M_i(u)$ has a total of $|S(u,i)|\cdot |P(u,i)|=O(n)$ messages, and a total of $O(n)$ bits. 
In step $i$, each node $u$ has two $O(n)$-bit input messages $\Minput(u) = \{M^{\circ-}_i(u),M^{\circ+}_i(u)\}$. Moreover, $u$ is a target of $O(n)$-bit input messages as well. Assuming that $\ID(u)[i]=0$ (i.e., $u=\Flip(u,i,0)$), we have $\Moutput(u)=\{M^{\circ-}_i(\Flip(v,i,0), M^{\circ-}_i(\Flip(v,i,1))\}$, and we have $\Moutput(u)=\{M^{\circ+}_i(\Flip(v,i,0), M^{\circ+}_i(\Flip(v,i,1))\}$, otherwise. 
Therefore, \Cref{thm:generalized_routing_main} guarantees that Proc. $\ExtMatchingTransmission$ is implemented correctly in $O(1)$ rounds.
\end{proof}

This completes the proof of \Cref{lem:det_mobile_high_noise_main} (i.e. \Cref{thm:informal_adaptive_deterministic_high_noise}).

    \subsection{Deterministic Algorithm for $\alpha=O(1/\sqrt{n})$}\label{sec:det-nonconstant}

In this section, our goal is to solve the $\AllToAllComm$ problem in $O(1)$ rounds for the case that $\alpha = O(1/\sqrt{n})$. By \Cref{lem:change_n_value}, we can assume w.l.o.g. that $\sqrt{n}$ is an integer. A very similar construction is provided in \cite{ABEGH19} (Theorem 1, therein) for the stochastic interactive coding setting. 

\begin{theorem}
    \label{thm:det_mobile_low_noise_main}
    Let $\alpha = O(1/\sqrt{n})$. There is a deterministic protocol solving $\AllToAllComm$ in the $\alpha$-ABD setting in $O(1)$ rounds, for bandwidth $B = 1$.
\end{theorem}

\smallskip
\noindent \textbf{Algorithm Description.}

Initially, each node $u$ holds the set of messages $\{m_{u,v}\}_{v \in V}$. Let $S_1,\ldots,S_{\sqrt{n}}$ be a partition of $V$ into consecutive segments of size $\sqrt{n}$, i.e. $S_1 = \{v_1,\ldots,v_{\sqrt{n}}\}, S_2 = \{v_{\sqrt{n}+1},\ldots,v_{2\sqrt{n}}\}, \ldots S_{\sqrt{n}} = \{v_{n-\sqrt{n}+1},\ldots,v_n\}$. Let $S_i[j]$ be the $j^{th}$ node in $S_i$, i.e. $S_i[j] = v_{(i-1) \cdot \sqrt{n}+j}$. We next describe our algorithm, which consists of two steps:

\paragraph*{Step 1:}
In this step, for each $i,j \in [\sqrt{n}]$, each node $S_i[j]$ learns the set of messages $M(S_i,S_j)$. This is done using $\ExtMatchingTransmission$ (See \Cref{lem:det_low_noise_step_1_route} for details).

\paragraph*{Step 2:} In this step, each node $v$ learns the set $M(S_i,\{v\})$ for all $i \in \sqrt{n}$. This is done using $\ExtMatchingTransmission$ (See \Cref{lem:det_low_noise_step_2_route} for details).

\paragraph*{Output:} In the previous step, each node $v$ learned $M(S_i,\{v\})$ for all $i \in [\sqrt{n}]$. We note that $M^{\circ}(V,\{v\}) = M^{\circ}(S_1,\{v\}) \circ \ldots \circ M^{\circ}(S_{\sqrt{n}},\{v\})$. Node $v$ outputs $M(V,\{v\})$, i.e. it outputs $m'_{u,v} = m_{u,v}$ for all $u \in V$. This concludes the description of the protocol. For a pictorial example of the algorithm for $n=9$, see \Cref{fig:det_low_noise}.

\begin{figure}
    \centering
    \includegraphics[width=0.8\textwidth]{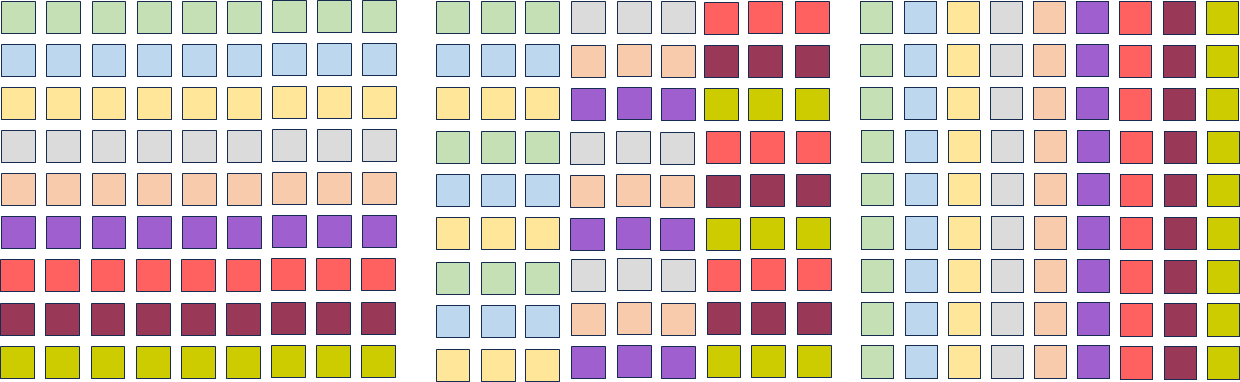}
    \caption{An illustration of the algorithm of \Cref{thm:det_mobile_low_noise_main} for $n=9$. Each matrix corresponds to the set of messages the nodes hold each of the three steps, where each column $i$ corresponds to the set of messages node $v_i$ holds at that step. Different colors of rectangles correspond to different targets. Initially each node has a message for each node. After Step 1, each of the three segments $S_i \in S_1,S_2,S_3$ collectively hold the set of messages $M(V,S_i)$. Finally, after Step 2, all nodes hold their target messages.}
    \label{fig:det_low_noise}
\end{figure}

\smallskip
\noindent \textbf{Analysis.}

\begin{lemma}
    \label{lem:det_low_noise_step_1_route}
    In Step $1$, each node $S_i[j]$ can learn $M(S_i,S_j)$ using $\ExtMatchingTransmission$ in $O(1)$ rounds.
\end{lemma}
\begin{proof}
    We use $\ExtMatchingTransmission$ where each node is the source and target of $O(\sqrt{n})$ super messages of size $O(\sqrt{n})$ bits.
    
    For each node $v$ we set $\Minput(v) = \{M^{\circ}(\{v\},S_j)\}_{j \in [\sqrt{n}]},$ and for each node $S_i[j]$, set 
    $\Moutput(S_i[j]) = \{M^{\circ}(\{v\},S_j)\}_{v \in S_i}.$ In particular node $S_i[j]$ learns the sets of messages $M(\{S_i[1]\},S_j),\ldots,M(\{S_i[\sqrt{n}]\},S_j)$, and can compute $M(S_i,S_j)$. Each node $v$ has $|\Minput(v)|,|\Moutput(v)| \leq \sqrt{n}$, and each super-message of size $O(\sqrt{n})$ bits. Therefore, \Cref{thm:generalized_routing_main} guarantees that Proc. $\ExtMatchingTransmission$ terminates correctly in $O(1)$ rounds.
\end{proof}

\begin{lemma}
    \label{lem:det_low_noise_step_2_route}
    In Step $2$, each node $v$ can learn $M(S_i,\{v\})$ for all $i \in [\sqrt{n}]$ using $\ExtMatchingTransmission$ in $O(1)$ rounds.
\end{lemma}
\begin{proof}
    We use $\ExtMatchingTransmission$ where each node is the source and target of $O(\sqrt{n})$ super messages of size $O(\sqrt{n})$ bits.
    
    For each node $S_i[j]$ we set $\Minput(S_i[j]) = \{M^{\circ}(S_i,\{S_j[\ell]\})\}_{\ell \in [\sqrt{n}]}.$ We note that the size of a super-message is the total size of the messages of $M(S_i,\{S_j[\ell]\})$, i.e. $O(\sqrt{n})$ bits. Node $v$ defines $\Moutput(v) = \{M^{\circ}(S_i,\{v\})\}_{i \in [\sqrt{n}]}$, i.e. it learns $M(S_1,\{v\}),\ldots,M(S_{\sqrt{n}},\{v\})$. Indeed, each node $v$ has $|\Minput(v)|,|\Moutput(v)| \leq \sqrt{n}$, and each super-message is of size $O(\sqrt{n})$. Therefore, \Cref{thm:generalized_routing_main} guarantees that Proc. $\ExtMatchingTransmission$ is terminates correctly in $O(1)$ rounds.
\end{proof}

\begin{lemma}
    \label{lem:det_low_noise_correctness}
    At the end of the procedure, each node $v$ outputs $m'_{u,v} = m_{u,v}$ for all $u \in V$.  
\end{lemma}
\begin{proof}
    By \Cref{lem:det_low_noise_step_2_route}, each node $v$ learns the set of messages $M(S_i ,\{v\})$ for all $i \in [\sqrt{n}]$. Since $S$ is a consecutive partition of $V$, then $M^{\circ}(V,\{v\}) = M^{\circ}(S_1,\{v\}) \circ \ldots \circ M^{\circ}(S_{\sqrt{n}},\{v\})$, it follows that the node $v$ learns each message of $M(V,\{v\})$, and outputs $m'_{u,v} = m_{u,v}$.
    
\end{proof}

This completes the proof of \Cref{thm:det_mobile_low_noise_main} (i.e. \Cref{thm:informal_adaptive_deterministic_low_noise}).

\paragraph{Acknowledgment.}
We thank Ran Gelles and Keren Hillel-Censor for discussions on LDCs and their applications in various resilient distributed algorithms. We also thank Ran Gelles for pointing out the relation between \Cref{thm:informal_adaptive_deterministic_low_noise} and \cite{ABEGH19}. Finally, we are very grateful to Orr Meir and Noga Ron-Zewi for advising on the internals of their LDC construction in \cite{KMRS17} (Lemma \ref{lem:LDC}).

    \bibliographystyle{plain}
    \bibliography{bibliography}
    
    \appendix

    \section{Missing Proofs for \Cref{sec:key_routing}}\label{app:routing}

    \subsection*{Derandomization of $(r,\delta)$-Cover Free Sets (Proof of \Cref{lem:cover_free_proof})}
    \label{sec:cover_free_lll_derandomization}
    In this section, we show how to locally compute $(r,\delta)$-cover free sets deterministically in an efficient manner, using a deterministic version of the Lovasz Local Lemma (LLL). We remind that the randomized construction is heavily inspired by one of the constructions of \cite{KRS99}. Before describing our construction, we give a brief overview of preliminaries regarding LLL and its algorithmic variants.
    
    \paragraph{Lovasz Local Lemma:}
    The setting of the Lovasz Local Lemma is the following:  
    
    Let $\mathcal{Y} = \{Y_1,\dots,Y_R\}$ be a set of mutually independent random variables receiving values from a universe $\Sigma$. Let $\mathcal{B} = B_1,\dots,B_L$ be a set of events (referred to as $\emph{bad events}$), each depending only on a subset of the random variables $\mathcal{Y}$. Denote by $\mathcal{Y}(B_i)$ the set of variables that event $B_i$ is dependent on. 

    We define the dependency graph $G$ on vertex set $\mathcal{B}$, where two vertices (events) $B_i,B_j$ share an edge if and only if $\mathcal{Y}(B_i) \cap \mathcal{Y}(B_j) \neq \emptyset$. In particular, the neighborhood of $B_i$ is defined as
    $$\mathrm{Neigh}(B_i) = \{B_j \stt \mathcal{Y}(B_i) \cap \mathcal{Y}(B_j) \neq \emptyset\} \setminus \{B_i\}.$$  Denote $d = \max_{B_i \in \mathcal{B}}  |\mathrm{Neigh}(B_i)|$ the maximum degree in the dependency graph. Finally, let $p_{\max} = \max_{B_i \in \mathcal{B}} \Pr(B_i)$ be the maximum probability of a bad event occurring.
    
    The seminal Lovasz Local Lemma \cite{EL74} guarantees us that if $ep_{\max}(d+1) < 1$ there exists an assignment of values to the random variables of $\mathcal{Y}$ such that no bad event $B_1,\dots,B_L$ occurs (where $e = 2.718\ldots$ denotes Euler's constant). We refer to such an assignment as a good assignment henceforth.

    \paragraph{Algorithmic Lovasz Local Lemma:}
    Algorithmically finding a good assignment for the variables of an LLL instance has been extensively studied. For our derandomization, we use the deterministic algorithm of Harris \cite{H23}.

    A partial expectation oracle (PEO) for an LLL instance is a deterministic sequential polynomial time algorithm (in terms of $|\mathcal{Y}|,|\mathcal{B}|,|\Sigma|$), that given a bad event $B_j$ and an assignment $Y_{i_1} = \sigma_1,\dots Y_{i_j}= \sigma_j$ of values to a subset of the variables of $\mathcal{Y}$, computes $\Pr(B_j \mid Y_{i_1} = \sigma_1,\dots Y_{i_j}= \sigma_j),$
    the probability of $B_j$ occurring conditioned on the partial assignment.
    
    The deterministic algorithm of Harris \cite{H23} finds a good assignment deterministically in polynomial sequential time, assuming a slightly stronger condition holds, and the existence of a PEO:

    \begin{lemma}[Deterministic Lovasz Local Lemma, \cite{H23} Corollary 1.2 (c.f. \cite{H23} Theorem 3.7(a))]
         If there exists a constant $\epsilon > 0$ such that $e p_{\mathrm{max}} d^{1+\epsilon} < 1$ and the bad events have a PEO, then there is a $\poly(|\mathcal{B}|,|\mathcal{Y}|,|\Sigma|)$-time deterministic sequential algorithm which finds an assignment of values to the random variables of $\mathcal{Y}$ such that no bad event $B_1,\dots,B_L$ occurs.
    \end{lemma}

    Finally, we give a brief overview of a well-studied distribution known as the poisson binomial distribution, which is needed for our proof. 
    \begin{definition}
        A random variable $S$ is said to have a poisson binomial distribution if there exist $I_1,\dots,I_\ell \in \{0,1\}$ mutually independent Bernoulli random variables, for some integer $\ell$, such that $S = \sum_{i=1}^\ell I_\ell$.   
    \end{definition}

    We remark that the variables $I_1,\dots,I_\ell$ need not have the same distribution, i.e. $\Pr(I_j = 1) = p_j$ for some $p_j \in [0,1]$.  

    \begin{lemma}[Immediate corollary from \cite{WS73} Lemma 1]
    \label{lem:poisson_binomial_alg}
        For a poisson binomial variable $S = \sum_{i=1}^\ell I_\ell$ and any integer $x$, computing $\Pr(S = x)$  can be performed sequentially deterministically in polynomial time in the encoding lengths of the variable probabilities $p_1,\dots,p_\ell$.\footnote{For the purposes of this section, the encoding length of a non-integer rational value is the encoding length of the numerator times the encoding length of the denominator, where we assume their minimality.}
    \end{lemma}
    
    We are now ready to prove \Cref{lem:cover_free_proof}. We split the proof into three parts: in the first part, we define an LLL instance where a bad event is synonymous to $|A_{i_0} \setminus \cup_{1 \leq j \leq r} A_{i_j}| < (1-\delta)|A_{i_0}|$ for some $\{i_0,\dots,i_r\} \in H$. Then we show that the LLL criteria $ep_{\max}d^{1.5} < 1$ holds. Finally, we show a partial expectation oracle (PEO) for this instance.

    \paragraph*{Defining the Random Variables and Bad Events.}
        
    Let $\gamma = \delta/4$. Therefore, $L = \lfloor N\gamma/(r+1) \rfloor$. Partition $[N]$ into consecutive sets $S_1,\dots,S_{L}$ of size $\lfloor \frac{r+1}{\gamma} \rfloor$ each (ignoring the remaining elements). Each set $A_1,\dots,A_m$ is an i.i.d. random set which contains exactly one uniformly random element from each of $S_j$, i.e. each element is chosen with probability $1/\lfloor \frac{r+1}{\gamma} \rfloor \leq 2/(\frac{r+1}{\gamma}) \leq \frac{2\gamma}{r+1}$, where the first inequality follows due to the fact that for any $x > 1$, $x \leq 2\lfloor x \rfloor$.  

    We define the random variable $Y_{i,j} \in [N]$ to be the unique element of $A_i \cap S_j$, and set 
    $$\mathcal{Y} = \{Y_{i,j} \mid 0 \leq i \leq r, j \in [L]\}.$$ 
    
    We notice that the random variables of $\mathcal{Y}$ are mutually independent. Next, we define the following LLL instance over the set of random variables $\mathcal{Y}$: For each tuple $(i_0,h) = (i_0,\{i_0,\dots,i_r\}) \in [m] \times H$, define a bad event $B(i_0,h)$ to be the event that $|A_{i_0} \setminus \cup_{1 \leq j \leq r} A_{i_j}| < (1-\delta)|A_{i_0}|$. We notice that there are at $(r+1)|H|$ bad events - one for each tuple $h \in H$, and choice of $i_0 \in h$.

    \paragraph*{Proving the LLL Criteria.} In this part we show that the condition $ep_{\max}d^{1.5} < 1$ holds.
    \begin{lemma}
        We have that $p_{\max} \leq \frac{1}{eN^2|H|^2}$.
    \end{lemma}
    \begin{proof}
        Fix a bad random event $B(i_0,h) \in \mathcal{B}$, where we denote $h = \{i_0,\dots,i_r\}$. By the union bound, the unique element $a \in A_{i_0} \cap S_j$ is contained in any of $A_{i_1},\dots,A_{i_r}$ with probability at most $r \frac{2\gamma}{r+1} \leq 2\gamma = \delta/2$. Let $I_j(i_0,\dots,i_r)$ be an indicator variable that this occurs. Then for $X(i_0,\dots,i_r) = \sum_{j=1}^{L} I_{j}(i_0,\dots,i_r)$, it holds that $E(X(i_0,\dots,i_r)) \leq (\delta/2) \cdot L$. We note that $I_1(i_0,\dots,i_r),\dots,I_{L}(i_0,\dots,i_r)$ are independent random variables. By choice of $c_2$, we get that $L = \lfloor N\gamma/(r+1) \rfloor \geq c'\log{N}$, for some sufficiently large $c' > 0$, and in particular we assume $c' \geq 36+24c_2$. Therefore by Chernoff, 

        $$\Pr(|A_{i_0} \setminus \bigcup_{j=1}^r A_{i_j}| < (1-\delta)|A_{i_0}|) = \Pr\left(X > \delta|A_{i_0}|\right) \leq e^{-c'\log{N}/12} \leq e^{-(3+2c_2)\log{N}}\leq \frac{1}{N^3|H|^2} \leq \frac{1}{eN^2|H|^2}.$$ 
    \end{proof}
    
    \begin{corollary}
        We have that $ep_{\max}d^{1.5} < 1$.
    \end{corollary}
    \begin{proof}
        Since the number of vertices in the dependency graph is $(r+1)|H| \leq N|H|$, in particular the maximum degree is at most $d \leq N|H|$. Combined with the fact that $p_{\max} \leq 1/(eN^2|H|^2)$, it follows that $ep_{\max}d^{1.5} < 1$. 
    \end{proof}
        
        \paragraph*{Showing a Partial Expectation Oracle.}        
        Finally, we show a PEO for this setting. Fix a bad event $B(i_0,h)$ for some $h = \{i_0,\dots,i_r\}$, and some partial assignment $F(\mathcal{Y})$ to the variables of $\mathcal{Y}$. For $j \in [L]$, denote 
        $$\mathrm{Fixed}(j) = \{i \in \{i_1,\dots,i_r\} \mid Y_{i,j} \text{ is assigned a value in } F(\mathcal{Y})\},$$
        $$\mathrm{FixedValues}(j) = \{\sigma \in [N] \mid Y_{i,j} \text{ is assigned the value $\sigma$ in } F(\mathcal{Y}) \text{ for some $i \in \{i_1,\dots,i_r\}$}\}.$$
        
        Recall that we defined $I_j(i_0,\dots,i_r)$ as the indicator (i.e. Bernoulli) variable that the unique element $a_j \in A_{i_0} \cap S_j$ is contained in any of $A_{i_1},\dots,A_{i_r}$. We start by showing an explicit formula of the probability $\Pr(I_j(i_0,\dots,i_r) = 1 \mid F(\mathcal{Y}))$ for any $j \in [L]$.
        For this, we split into cases:
        \begin{itemize}
            \item If $F(\mathcal{Y})$ assigns $Y_{0,j} = Y_{i,j} = \sigma$ for some $\sigma \in \Sigma$ and $i \in \{1,\dots,r\}$: then $\sigma \in A_{i_j} \cap A_{i_0}$, and in particular,
            $\Pr(I_j(i_0,\dots,i_r) = 1 \mid F(\mathcal{Y})) = 1$.
            \item  If $F(\mathcal{Y})$ assigns $Y_{0,j} = \sigma$ for some $\sigma \in [N]$, and for any $i \in \{1,\dots,r\}$ either it assigns no value, or assigns a value $Y_{i,j} \neq \sigma$. Then
            $$\Pr(I_j(i_0,\dots,i_r) = 1 \mid F(\mathcal{Y})) = \Pr(\sigma \in \bigcup_{j=1}^r A_{i_j} \mid F(\mathcal{Y})) =  1-(1/\lfloor \frac{r+1}{\gamma}\rfloor)^{r-|\mathrm{Fixed}(j)|}.$$
            \item If $F(\mathcal{Y})$ does not assign $Y_{0,j}$, we consider the following equivalent random process: first, we choose $Y_{0,j}$, and only then the rest of $Y_{i,j}$ not assigned by $F(\mathcal{Y})$. If $Y_{0,j}$ chooses a value from $Y_{0,j} \in \mathrm{FixedValues}(j)$, we are exactly in the first case, and if it chooses $Y_{0,j} \notin \mathrm{FixedValues}(j)$, we are exactly in the second case. Hence, the probability is equal to:
            $$\Pr(I_j(i_0,\dots,i_r) = 1 \mid F(\mathcal{Y})) = \frac{|\mathrm{FixedValues}(j)|}{\lfloor \frac{r+1}{\gamma}\rfloor} + (1-\frac{|\mathrm{FixedValues}(j)|}{\lfloor \frac{r+1}{\gamma}\rfloor})(1-(1/\lfloor\frac{r+1}{\gamma}\rfloor)^{r-|\mathrm{Fixed}(j)|}).$$
        \end{itemize}

        Indeed, this probability can be computed exactly in polynomial time, and can be encoded (in a numerator-demoninator encoding) using $\poly(N)$ bits. We remind that $X(i_0,\dots,i_r) = \sum_{j=1}^{L} I_{j}(i_0,\dots,i_r)$, and therefore is the sum of independent Bernoulli variables, i.e.  $X(i_0,\dots,i_r)$ has a poisson-binomial distribution. Hence, the probability that it receives any given value can be computed in polynomial time (\Cref{lem:poisson_binomial_alg}). In particular,
        $$\Pr(X(i_0,\dots,i_r) > \delta |A_{i_0}| \mid F(\mathcal{Y})) = \sum_{\ell = \lfloor \delta \cdot L  \rfloor +1}^{L} \Pr(X(i_0,\dots,i_r) = \ell \mid F(\mathcal{Y}))$$ can be computed in polynomial time, and the claim follows.

    \section{Missing Proofs for  \Cref{sec:adaptive}}\label{app:adaptive}
    
    \APPENDRandomPartition
    \APPENDADAPCORRECT

\end{document}